\documentclass[12pt,letterpaper]{article}

\usepackage{xcolor, colortbl}
\usepackage{amssymb, amsmath, amsthm}
\usepackage{bm}
\usepackage{graphicx}
\usepackage{natbib}
\usepackage{float}
\usepackage{setspace}
\usepackage{algorithm}

\definecolor{lightgray}{gray}{0.9}

\newtheorem{lemma}{Lemma}
\newtheorem{proposition}{Proposition}
\newtheorem{theorem}{Theorem}
\newtheorem{corollary}{Corollary}
\newtheorem{assumption}{Assumption}

\newcommand{\di}{\text{d}}
\newcommand{\bh}{\mathbf{h}}
\newcommand{\bmm}{\mathbf{m}}

\newcommand{\by}{\mathbf{y}}
\newcommand{\bx}{\mathbf{x}}
\newcommand{\bSigma}{\bm{\Sigma}}
\newcommand{\bB}{\mathbf{B}}

\newcommand{\bK}{\mathbf{K}}
\newcommand{\bQ}{\mathbf{Q}}

\newcommand{\bone}{\mathbf{1}}
\newcommand{\distn}[1]{\mathcal{#1}}
\newcommand{\Em}{\mathbb E}
\newcommand{\Pm}{\mathbb P}
\newcommand{\gvn}{\,|\,}
\newcommand{\e}{\text{e}}
\newcommand{\Var}{\text{Var}}
\newcommand{\diag}{\text{diag}}
\renewcommand{\epsilon}{\varepsilon}
\renewcommand{\hat}{\widehat}
\renewcommand{\tilde}{\widetilde}
\renewcommand{\leq}{\leqslant}
\renewcommand{\geq}{\geqslant}
\newcommand{\matlab}{\mathrm{M}\mathrm{{\scriptstyle ATLAB}}}

\begin{document}

\title{Exact Rejection Sampling for Non-Gaussian State Space Models}

\author{Joshua C.C. Chan \\
 {\small Purdue University}
}

\date{August 2026}

\maketitle

\onehalfspacing

\begin{abstract}

\noindent Rejection sampling requires a proposal that dominates the target by a known constant, generally unavailable for non-Gaussian state space models. We construct such a proposal for the latent state path, yielding independent exact smoothing draws and an unbiased likelihood estimator whose relative variance is at most $1/p-1$ per draw at acceptance probability $p$. The method covers scalar states with affine Gaussian dynamics and log-concave observation densities, including multivariate observations. Transition twisting makes the log target-to-proposal ratio separable, and tangent-line twists make each term nonpositive, producing an attained, sharp dominating constant. With a companding node placement, the accumulated envelope error is $O(T/G^2)$ for a sample of length $T$ with $G$ nodes per date, so $G\propto\sqrt{T}$ keeps acceptance bounded away from zero; for stochastic volatility, the required conditions hold almost surely. A simpler mode-centered grid shows the same scaling empirically. At $T=2{,}000$, acceptance is $75\%$, versus roughly $10^{-16}$ for the Gaussian envelope.

\smallskip


\smallskip


\end{abstract}

\thispagestyle{empty}

\newpage

\section{Introduction} \label{s:intro}

Latent state variables are central to many models in macroeconomics and finance. Stochastic volatility, unobserved components, count and duration models, among others, require simulation from the conditional distribution of a latent state path. When the model is linear and Gaussian, the path can be drawn exactly and cheaply using standard simulation smoothers.\footnote{These include the smoothers of \citet{CK94}, \citet{dS95}, \citet{DK02}, \citet{CJ09} and \citet{MMP11}.} Outside that class, state simulation typically relies on Metropolis--Hastings updates, importance sampling, particle methods, or approximations that replace the measurement density by a mixture of normals so that linear Gaussian machinery applies \citep{KSC98, OCSN07}. These methods can be highly accurate, but generally do not provide independent exact draws of the entire state path along with a computable global bound on the target-to-proposal ratio.

This paper develops an exact rejection sampler for a class of non-Gaussian state space models with a scalar latent state, affine Gaussian dynamics and a finite, continuously differentiable, concave observation log-density. The observation may be multivariate. Throughout, ``exact'' refers to the conditional state draw at a fixed parameter vector: an accepted path is distributed exactly from the joint smoothing distribution of the intended model. There is no discretization or approximation-model error, and the state draw requires neither burn-in nor a Markov-chain correction. The exactness claim concerns the algorithm in exact arithmetic; the implementation is subject to floating-point roundoff. 

The main obstacle to rejection sampling the entire path is to construct a proposal that dominates the target up to a \emph{known} constant. For a $T$-dimensional state path, no practically useful constant has generally been available. A Gaussian envelope centered at the posterior mode provides a computable bound, but its acceptance probability decays exponentially with $T$; for the observation densities considered here, changing the Gaussian covariance does not in general resolve the problem. Existing non-Gaussian simulation smoothers instead construct accurate approximations to the smoothing distribution. Their target-to-proposal ratio may be bounded, but the value of that bound is not computed, whereas rejection sampling requires it.

Our construction makes the global bound tractable in two steps. First, twisting the transition densities makes the log target-to-proposal ratio, apart from an additive constant, a \emph{sum of univariate functions} of the states. Its $T$-dimensional supremum therefore separates into $T$ scalar suprema. Second, we take each twist to be the minimum of tangent lines to a concave backward function. The resulting tangent envelope dominates that function and touches it at the nodes, so every scalar supremum is exactly zero. At the same time, each twisted transition remains a finite mixture of truncated normals with closed-form normalizing constants and draws. The global dominating constant is then simply the normalizing constant of the first twisted density, computed as a by-product of the backward recursion.

The paper makes two main contributions. The first is an exact simulation smoother for the model class described above. We prove that the tangent-twisted proposal dominates the unnormalized smoothing density everywhere on $\mathbb{R}^T$, with a dominating constant that is attained and is the smallest admissible for the proposal. Rejection sampling therefore produces independent draws from the exact joint smoothing distribution. To our knowledge, this is the first practical rejection sampler for the joint smoothing distribution of a continuous, unbounded state process with Gaussian transitions for which the global dominating constant is both computable and sharp.

We also characterize how the computational cost of exactness scales with the path length. In nondegenerate cases, a fixed per-period envelope error accumulates along the path and causes acceptance to deteriorate exponentially in $T$. For the tangent construction, the per-period error is instead of order $G^{-2}$ in the number of nodes $G$, so the accumulated error, and hence $-\log$ of the acceptance probability, is of order $T/G^2$. Increasing $G$ at rate $\sqrt{T}$ therefore keeps the acceptance probability bounded away from zero. Theorem~\ref{thm:scaling} establishes this result under uniform curvature and proposal-tail conditions using a companding rule for node placement; for the stochastic volatility model, these conditions hold almost surely under the data generating process. The same square-root rate is recommended for the discretization filter of \citet{Farmer2021}, where it instead controls the error from replacing the continuous state process by a finite Markov chain.

The second contribution is likelihood evaluation with importance weights bounded by the same constant. The construction gives a nonnegative unbiased likelihood estimator whose variance is finite at every parameter value, with relative variance bounded in terms of the rejection acceptance probability rather than diagnosed from realized weights. Finite variance cannot be taken for granted for Gaussian proposals in this literature: \citet*{KoopmanShephardCreal2009} argue that it should be tested rather than assumed, and for the Gaussian approximation at the mode the relevant test fails throughout the stochastic volatility experiments and at all six empirical fits reported below. The bound also provides an exact reference for evaluating approximate smoothers, a computable uniform-ergodicity bound for the associated independence sampler, and a building block for pseudo-marginal and marginal likelihood estimation, as used for large vector autoregressions with stochastic volatility by \citet{chan23JE} and \citet*{CYZ26}.

The Monte Carlo results show that these guarantees remain useful for long state paths. For the stochastic volatility model, the rejection sampler accepts $75\%$ of proposals at $T=2{,}000$, compared with fewer than one in $10^{16}$ for the Gaussian envelope, with similarly high acceptance across seven other observation densities. Taking $G$ proportional to $\sqrt{T}$ keeps acceptance approximately constant across a sixteenfold increase in sample size, at about $0.87$ for the practical grid used in the implementation; the companding grid of Theorem~\ref{thm:scaling} shows the same scaling at a uniformly lower level. Under a common computing-time budget, the likelihood estimator is two to three thousand times more efficient than a bootstrap particle filter and comparable to the efficient importance sampler of \citet{RichardZhang2007}. Against importance sampling from the Gaussian approximation at the mode no such ratio is well defined, because those weights have infinite variance in every design we consider. What distinguishes the tangent estimator is therefore not the margin but that it attains this efficiency with weights bounded by a constant the algorithm computes. Exact draws also provide a direct benchmark for quantifying approximation error in standard state smoothers.

We illustrate the method with the stochastic conditional duration model of \citet{BauwensVeredas2004}. Marginal likelihoods strongly favor gamma over exponential errors in both series considered. Imposing exponential errors forces the latent state to absorb dispersion that would otherwise be captured by the error distribution, increasing the estimated state innovation variance several-fold.

The construction connects to several strands of the literature. The twisted transition densities are Doob $h$-transforms of the type used in sequential Monte Carlo \citep{WhiteleyLee2014, GuarnieroJohansenLee2017, HengBishopDeligiannidisDoucet2020}. The backward recursion is closely related to the efficient importance sampling method of \citet{RichardZhang2007} and the HESSIAN method of \citet{McCausland2012}, and more broadly to the Gaussian importance smoothers of \citet{DanielssonRichard1993}, \citet{SP97} and \citet{DurbinKoopman1997}. These methods construct accurate global proposals for importance sampling or Metropolis--Hastings, but do not supply the value of a global bound over the entire state path. The tangent construction is the path-space analogue of the piecewise linear upper hull used for univariate log-concave densities by \citet{Devroye1986} and, with adaptive refinement of the node set, by \citet{GilksWild1992}.

The rest of the paper is organized as follows. Section~\ref{s:method} establishes exactness of the tangent-twisted rejection sampler, and Section~\ref{s:implementation} presents the implementation and scaling theory. Section~\ref{s:uses} develops uses beyond exact state simulation, and Section~\ref{s:related} relates the construction to existing smoothers. Sections~\ref{s:MC} and~\ref{s:app} report the Monte Carlo evidence and the application, and Section~\ref{s:conclusion} concludes.

\section{Exact Rejection Sampling for the State Path} \label{s:method}

This section sets out the modeling framework and the assumption maintained throughout, and shows why the obvious Gaussian envelope cannot work. It then builds the proposal in two steps: twisting the transition densities makes the log target-to-proposal ratio separable across dates, and tangent-line twists make every term of that sum nonpositive. The main result is Theorem~\ref{thm:exact}: the algorithm is exact and its dominating constant is attained and smallest admissible.

\subsection{The Modeling Framework} \label{ss:model}

Let $\bh=(h_1,\ldots,h_T)'$ denote a scalar latent state sequence and $\by=(\by_1',\ldots,\by_T')'$ the observations, where each $\by_t$ may be of any dimension. The observations are conditionally independent given the states, and the state evolves as a Markov chain with Gaussian transitions,
\begin{equation} \label{eq:model}
	\log p(\by_t \gvn h_t = u) = \ell_t(u), \qquad
	h_1 \sim \distn{N}(\mu_1,\omega_1), \qquad
	(h_t \gvn h_{t-1}) \sim \distn{N}\left(a_t(h_{t-1}),\omega_t\right),
\end{equation}
for $t=2,\ldots,T$. Write $p_1(\cdot)$ for the density of $h_1$ and $p_t(\cdot \gvn u)$ for the transition density. The functions $\ell_t$ and $a_t$ depend on static parameters, which are suppressed from the conditioning set to keep the notation uncluttered. The unnormalized posterior of the state path and its normalizing constant are
\begin{equation} \label{eq:target}
	\gamma(\bh) = p_1(h_1)\e^{\ell_1(h_1)}\prod_{t=2}^{T}p_t(h_t \gvn h_{t-1})\e^{\ell_t(h_t)},
	\qquad Z = \int_{\mathbb{R}^T}\gamma(\bh)\,\di\bh,
\end{equation}
so that the target is $p(\bh \gvn \by) = \gamma(\bh)/Z$. The objective is to obtain independent draws from $p(\bh \gvn \by)$ exactly. 

The leading example is the stochastic volatility model $y_t = \e^{h_t/2}\epsilon_t$ with $\epsilon_t\sim\distn{N}(0,1)$ and $h_t = \mu+\phi(h_{t-1}-\mu)+\sigma_h\eta_t$, $\eta_t\sim\distn{N}(0,1)$, where $|\phi|<1$ and the state is initialized from its stationary distribution. In the notation of \eqref{eq:model}, $a_t(u) = \mu+\phi(u-\mu)$, $\omega_t=\sigma^2_h$, $\mu_1=\mu$, $\omega_1 = \sigma^2_h/(1-\phi^2)$ and $\ell_t(u) = -\frac{1}{2}(u+y_t^2\e^{-u})-\frac{1}{2}\log(2\pi)$. The restriction $|\phi|<1$ is what makes $\omega_1$ positive and finite; a nonstationary state is covered by \eqref{eq:model} under any other initialization with $\omega_1>0$.

We maintain two conditions throughout. The first restricts the state equation and is used only to establish a log-concavity property in Lemma~\ref{lem:concave}; the second restricts the observation densities.

\begin{assumption} \label{as:model}
The following hold.
\begin{enumerate}
	\item[(i)] For each $t\geq 2$ the conditional mean $a_t(u) = \alpha_t+\beta_tu$ is affine in $u$ and the conditional variance $\omega_t>0$ does not depend on the state; and $\omega_1>0$.
	\item[(ii)] For each $t$ the observation log-density $\ell_t:\mathbb{R}\rightarrow\mathbb{R}$ is finite, continuously differentiable and concave.
\end{enumerate}
\end{assumption}

Assumption~\ref{as:model}(ii) accommodates a wide range of observation densities. For the stochastic volatility model, $\ell_t''(u)=-\frac{1}{2}y_t^2\e^{-u}\leq 0$. For the stochastic volatility in mean model $y_t = \alpha\e^{h_t}+\e^{h_t/2}\epsilon_t$ of \citet{KoopmanHolUspensky2002}, $\ell_t''(u) = -\frac{1}{2}y_t^2\e^{-u}-\frac{1}{2}\alpha^2\e^{u}\leq0$. For Student's $t$ measurement error with $\nu$ degrees of freedom, $\ell_t''(u) = -\frac{1}{2}(\nu+1)z/(1+z)^2 \leq0$, where $z=y_t^2\e^{-u}/\nu$, with equality only if $y_t=0$. For count data with a latent log-intensity, $y_t$ Poisson with mean $\e^{h_t}$, $\ell_t''(u)=-\e^{u}<0$. For a dynamic probit model with $\Pm(y_t=1)=\Phi(h_t)$, concavity of $\ell_t$ follows from the log-concavity of the normal distribution function. For the stochastic conditional duration model $y_t = \e^{h_t}\epsilon_t$ of \citet{BauwensVeredas2004} with exponential errors, $\ell_t(u) = -u-y_t\e^{-u}$ and $\ell_t''(u) = -y_t\e^{-u}<0$, and the fixed-shape gamma and Weibull versions have the same linear-minus-exponential form.

A further class deserves separate mention, because it shows that the restriction to a scalar state places no restriction on the dimension of the observation. In the common stochastic volatility model of \citet{CCM16}, $\by_t = \bB'\bx_t+\e^{h_t/2}\bm{\epsilon}_t$ with $\bm{\epsilon}_t\sim\distn{N}(\mathbf{0},\bSigma)$, the observation log-density conditional on $\bB$ and $\bSigma$ is
\[
	\ell_t(u) = -\frac{n}{2}u-\frac{1}{2}\e^{-u}(\by_t-\bB'\bx_t)'\bSigma^{-1}(\by_t-\bB'\bx_t)+\text{constant},
\]
which is concave in $u$ for every $n$. The same holds for a one-factor dynamic count model in which $y_{it}$ is Poisson with mean $\e^{\alpha_i+\lambda_ih_t}$ for $i=1,\ldots,n$, since $\ell_t$ is then a sum of $n$ concave functions of $u$. 

Assumption~\ref{as:model}(ii) requires concavity but not strict concavity. This matters in practice: in the stochastic volatility model an observation recorded as exactly zero gives $\ell_t(u)=-u/2-\frac{1}{2}\log(2\pi)$, which is affine rather than strictly concave.

\subsection{Why the Obvious Envelopes Fail} \label{ss:gaussian}

Before developing the proposal it is useful to see precisely why the obvious candidate fails. Let $\bQ$ denote the prior precision matrix of $\bh$ implied by \eqref{eq:model}, which is tridiagonal, and let $\bmm$ denote the posterior mode, with $t$th element $m_t$. Since each $\ell_t$ is concave, the tangent line at $m_t$ dominates it, and hence $\gamma(\bh) \leq \e^{M_G}q_G(\bh)$, where $M_G = \sum_{t=1}^{T}\ell_t(m_t)-\tfrac{1}{2}(\bmm-\bm{\alpha})'\bQ(\bmm-\bm{\alpha})$, $\bm{\alpha}$ collects the prior means, and $q_G$ is the $\distn{N}(\bmm,\bQ^{-1})$ density including its normalizing constant. 

The bound is attained at $\bh=\bmm$, so the dominating constant is exact and the acceptance probability is $Z\e^{-M_G}$, for which a Laplace approximation gives $\e^{-\kappa}$ with $\kappa=\frac{1}{2}(\log|\bK|-\log|\bQ|)\geq0$ and $\bK = \bQ+\diag\{-\ell_1''(m_1),\ldots,-\ell_T''(m_T)\}$. When the curvature $-\ell_t''(m_t)$ is bounded away from zero for a positive fraction of the sample and the eigenvalues of $\bQ$ are bounded above uniformly in $T$, as they are under Assumption~\ref{as:model}(i) with fixed state-equation parameters, $\kappa$ grows linearly in $T$ and the acceptance probability decays exponentially. 

Whether this deficiency can be repaired by changing the Gaussian covariance matrix depends on the tails of the observation density. Along the ray $\bh=\bmm+r s_j\mathbf{e}_j$ for a coordinate $j$ and a direction $s_j\in\{-1,1\}$, if $\ell_j(m_j+rs_j) = o(r^2)$ as $r\rightarrow\infty$, then $\log\gamma \sim -\frac{1}{2}r^2\bQ_{jj}$, and boundedness of $\log\gamma-\log q$ for a Gaussian proposal with precision $\bQ_q$ requires $(\bQ_q)_{jj} \leq \bQ_{jj}$: no Gaussian envelope can then add diagonal curvature in coordinate $j$, which is exactly the curvature that the observation density contributes and that $\log|\bK|-\log|\bQ|$ measures. The tail condition holds for the standard stochastic volatility model, where $\ell_t(u)\rightarrow -u/2$ as $u\rightarrow\infty$, and for the Student's $t$, Poisson and dynamic probit specifications discussed earlier.

A second route to a bounded weight avoids concavity altogether: the defensive importance sampling of \citet{Hesterberg1995} mixes the proposal with the state prior $p_0$, and taking $q_{\text{def}} = (1-\lambda)p_0+\lambda q$ gives $\gamma(\bh)/q_{\text{def}}(\bh) \leq (1-\lambda)^{-1}\exp\{\sum_t\sup_u\ell_t(u)\}$ from $T$ scalar maximizations. Two difficulties arise. The bound is finite only if every $\sup_u\ell_t(u)$ is finite, which Assumption~\ref{as:model} does not require: the affine $\ell_t$ produced by an observation recorded as exactly zero is unbounded above. When all $T$ suprema are finite, the bound is still too loose, because the implied acceptance probability is $(1-\lambda)\exp[-\sum_t\{\sup_u\ell_t(u)-\log p(\by_t \gvn \by_1,\ldots,\by_{t-1})\}]$, in which every summand is nonnegative and vanishes only if $\ell_t$ is flat where the predictive state distribution puts mass. Whenever the state matters, in the sense that the average of these gaps is bounded away from zero, the acceptance probability again decays exponentially in $T$.

\subsection{Separability and Tangent Twisting} \label{ss:separability}

The proposal is built by twisting the transition densities of the state equation: each is multiplied by a positive function of the current state and renormalized, so that the result remains a Markov chain that can be simulated forward in one pass. Let $\psi_1,\ldots,\psi_T$ be functions on $\mathbb{R}$ for which the integrals below are finite, and define
\[
	\hat C_1 = \int_{\mathbb{R}}p_1(v)\e^{\psi_1(v)}\,\di v,
	\qquad
	\hat C_t(u) = \int_{\mathbb{R}}p_t(v \gvn u)\e^{\psi_t(v)}\,\di v, \quad t=2,\ldots,T.
\]
The twisted proposal is the Markov chain with initial and transition densities
\begin{equation} \label{eq:proposal}
	q_1(h_1) = \frac{p_1(h_1)\e^{\psi_1(h_1)}}{\hat C_1},
	\qquad
	q_t(h_t \gvn h_{t-1}) = \frac{p_t(h_t \gvn h_{t-1})\e^{\psi_t(h_t)}}{\hat C_t(h_{t-1})},
\end{equation}
and joint density $q(\bh) = q_1(h_1)\prod_{t=2}^{T}q_t(h_t \gvn h_{t-1})$. Define the backward functions
\begin{equation} \label{eq:f}
	f_T = \ell_T, \qquad f_t = \ell_t + \log\hat C_{t+1}, \quad t=1,\ldots,T-1,
\end{equation}
and the residuals $d_t = f_t-\psi_t$.

\begin{lemma}[Separability] \label{lem:sep}
For every $\bh\in\mathbb{R}^T$,
\begin{equation} \label{eq:ratio}
	\frac{\gamma(\bh)}{q(\bh)} = \hat C_1\exp\left\{\sum_{t=1}^{T}d_t(h_t)\right\},
\end{equation}
and consequently $\sup_{\bh\in\mathbb{R}^T}\sum_{t=1}^{T}d_t(h_t) = \sum_{t=1}^{T}\sup_{u\in\mathbb{R}}d_t(u)$.
\end{lemma}

\begin{proof}
Dividing \eqref{eq:target} by the joint density implied by \eqref{eq:proposal}, every transition density cancels and
\begin{align*}
	\frac{\gamma(\bh)}{q(\bh)}
	&= \hat C_1\prod_{t=1}^{T}\e^{\ell_t(h_t)-\psi_t(h_t)}\prod_{t=2}^{T}\hat C_t(h_{t-1})\\
	&= \hat C_1\exp\left\{\sum_{t=1}^{T}\left[\ell_t(h_t)-\psi_t(h_t)\right]+\sum_{t=1}^{T-1}\log \hat C_{t+1}(h_t)\right\}.
\end{align*}
Collecting the terms that involve $h_t$ and using \eqref{eq:f} gives \eqref{eq:ratio}. The second claim follows because the summands depend on disjoint coordinates of $\bh$.
\end{proof}

The identity \eqref{eq:ratio} is the standard weight decomposition for twisted proposals, implicit in \citet{RichardZhang2007}. What matters here is its consequence: a supremum over the $T$-dimensional path is a sum of $T$ scalar suprema, whatever the twisting functions. The optimal twist $\psi_t=f_t$ makes every residual vanish, so that $\hat C_1=Z$ and every proposal is accepted, but it is not available in closed form. We instead choose $\psi_t$ so that each residual is nonpositive while each $\hat C_t$ remains analytic. Specifically, for each $t$ let $V_t=\{v_{t,1}<v_{t,2}<\cdots<v_{t,G_t}\}$ be a finite nonempty set of nodes, so that $G_t\geq1$, and define
\begin{equation} \label{eq:tangent}
	\psi_t(u) = \min_{1\leq j\leq G_t}\left\{f_t(v_{t,j})+f_t'(v_{t,j})(u-v_{t,j})\right\}.
\end{equation}
For a fixed concave function, \eqref{eq:tangent} is the upper hull of univariate adaptive rejection sampling \citep{GilksWild1992}; what is new here is that it is applied to functions the construction itself generates. The definition is recursive: $\psi_T$ is constructed from $f_T=\ell_T$, then $\hat C_T$ determines $f_{T-1}$ through \eqref{eq:f}, which determines $\psi_{T-1}$, and so on down to $t=1$. Each step integrates the twisted density of the following period, so the construction is well posed only if concavity and differentiability are preserved under that integration; Lemma~\ref{lem:concave} establishes that they are, and records the domination property on which Theorem~\ref{thm:exact} rests. Its proof is in \ref{app:proofs}.

\begin{lemma}[Concavity and domination] \label{lem:concave}
Under Assumption~\ref{as:model}, for every $t=1,\ldots,T$ the following hold.
\begin{enumerate}
	\item[(i)] $f_t$ is finite, concave and continuously differentiable on $\mathbb{R}$, so that \eqref{eq:tangent} is well defined.
	\item[(ii)] $\psi_t$ is concave and piecewise linear, with $\psi_t(u)\geq f_t(u)$ for every $u\in\mathbb{R}$ and $\psi_t(v)=f_t(v)$ for every $v\in V_t$.
	\item[(iii)] For $t\geq 2$, $0<\hat C_t(u)<\infty$ for every $u$, and $\log\hat C_t$ is concave and continuously differentiable; and $0<\hat C_1<\infty$.
\end{enumerate}
\end{lemma}

Part~(ii) is also what makes the twisted kernel tractable. Where $\psi_t$ is affine, multiplying the Gaussian transition density by $\e^{\psi_t}$ leaves a normal density with the same variance and a shifted mean, truncated to that piece. Hence $q_t(\cdot \gvn h_{t-1})$ is a mixture of truncated normals, and Section~\ref{ss:pieces} obtains its weights and $\hat C_t$ in closed form.

\subsection{Exactness and Bounded Weights} \label{ss:exact}

The two lemmas now combine: each residual $d_t=f_t-\psi_t$ is nonpositive by Lemma~\ref{lem:concave}(ii) and vanishes on the nodes, so the sum in \eqref{eq:ratio} is at most zero and the bound is attained.

\begin{theorem}[Exactness and the smallest dominating constant] \label{thm:exact}
Let Assumption~\ref{as:model} hold and let $\psi_t$ be given by \eqref{eq:tangent}. Then the following hold.
\begin{enumerate}
	\item[(i)] $q$ is a probability density on $\mathbb{R}^T$.
	\item[(ii)] $\gamma(\bh)\leq \hat C_1 q(\bh)$ for every $\bh\in\mathbb{R}^T$, with equality whenever $h_t\in V_t$ for every $t$; hence $\hat C_1 = \sup_{\bh}\gamma(\bh)/q(\bh)$, and no smaller dominating constant is admissible for this proposal.
	\item[(iii)] Let $\bh\sim q$ and $U\sim\distn{U}(0,1)$ be independent, and accept $\bh$ if
	\begin{equation} \label{eq:accept}
		\log U \leq \sum_{t=1}^{T}\left[f_t(h_t)-\psi_t(h_t)\right].
	\end{equation}
	Then the conditional distribution of $\bh$ given acceptance is exactly $p(\bh \gvn \by)$, and a proposal is accepted with probability $Z/\hat C_1$.
\end{enumerate}
\end{theorem}

The proof is in \ref{app:proofs}. Theorem~\ref{thm:exact} supplies what has been missing in this setting: a dominating constant that is known. It requires no numerical maximization, no safety factor and no search over the tails, because the supremum of each residual is zero by construction and is attained at the nodes. The theorem holds for any finite nonempty node sets. Exactness therefore does not depend on the accuracy of the posterior mode: poor placement lowers the acceptance probability but cannot invalidate the bound.

Theorem~\ref{thm:exact} also bounds the importance weights, which is what makes the twisted proposal useful beyond rejection sampling; the proof of the corollary is in \ref{app:proofs}.

\begin{corollary}[Bounded importance weights] \label{cor:is}
Let the conditions of Theorem~\ref{thm:exact} hold and let $w(\bh) = \gamma(\bh)/q(\bh)$. Then $0<w(\bh)\leq\hat C_1$ for every $\bh$, $\Em_q\{w(\bh)\}=Z$, and
\[
	\frac{\Var_q\{w(\bh)\}}{\left[\Em_q\{w(\bh)\}\right]^2} \ \leq\ \frac{\hat C_1}{Z}-1 \ =\ \frac{1}{p}-1,
\]
where $p$ is the acceptance probability of Theorem~\ref{thm:exact}. Hence for $\bh^{(1)},\ldots,\bh^{(M)}$ drawn independently from $q$, the estimator $\hat Z = M^{-1}\sum_{m}w(\bh^{(m)})$ is unbiased for $Z$ with relative variance at most $(1/p-1)/M$.
\end{corollary}

\section{Implementation and Scaling} \label{s:implementation}

The construction of Section~\ref{s:method} is useful only if the twisted densities can be sampled and their normalizing constants evaluated. This section shows that both are available in closed form, because a piecewise linear twist leaves each twisted transition a mixture of truncated normals. It then sets out how the nodes are chosen, how many of them the sample size calls for, and states the algorithm, closing with the numerical precautions that an implementation requires.

\subsection{Closed-Form Expressions} \label{ss:pieces}

Write $g_{t,j}=f_t'(v_{t,j})$ for the slope of the $j$th tangent line and $\alpha_{t,j}=f_t(v_{t,j})-g_{t,j}v_{t,j}$ for its intercept, so that the line is $u\mapsto \alpha_{t,j}+g_{t,j}u$. The following lemma records the interval on which each line is active; the intersection points and the interval rule are equations (1) and (2) of \citet{GilksWild1992}, restated for a general strictly concave $f_t$ and with the interlacing $v_{t,j}<z_{t,j}<v_{t,j+1}$ established rather than assumed. The proof is in \ref{app:proofs}.

\begin{lemma}[Tangent intervals] \label{lem:intervals}
Suppose in addition that $f_t$ is strictly concave. Then $g_{t,1}>g_{t,2}>\cdots>g_{t,G_t}$, the intersection points
\[
	z_{t,j} = \frac{\alpha_{t,j+1}-\alpha_{t,j}}{g_{t,j}-g_{t,j+1}}, \qquad j=1,\ldots,G_t-1,
\]
satisfy $v_{t,j}<z_{t,j}<v_{t,j+1}$ and hence $z_{t,1}<z_{t,2}<\cdots<z_{t,G_t-1}$, and
$\psi_t(v) = \alpha_{t,j}+g_{t,j}v$ for $v\in(z_{t,j-1},z_{t,j}]$, with the conventions $z_{t,0}=-\infty$ and $z_{t,G_t}=+\infty$.
\end{lemma}

When $f_t$ is concave but not strictly so, distinct nodes can share a tangent slope and the intersection $z_{t,j}$ is undefined. Equal slopes at two nodes imply that $f_t$ is affine between them, so the two tangent lines coincide rather than merely running parallel. The implementation therefore merges coincident lines and keeps one representative, then retains only the lines that appear on the lower envelope of the remaining family, discarding any that is nowhere the minimum, and forms intersections between adjacent distinct slopes. No slope is altered, so every retained line remains a tangent and $\psi_t\geq f_t$ is preserved exactly. When $f_t$ is affine a single line survives and the twisted transition reduces to a single untruncated normal, which is the case used as an implementation check in Section~\ref{s:MC}. After merging coincident lines and removing inactive lines, let $\mathcal{A}_t$ index the retained distinct lines. In the formulas below, sums over mixture components are understood to run over $j\in\mathcal{A}_t$. Under strict concavity, $\mathcal{A}_t=\{1,\ldots,G_t\}$.

On each of these intervals the twisted density is normal. Completing the square,
\[
	p_t(v \gvn u)\e^{\psi_t(v)}
	= \exp\left\{\alpha_{t,j}+g_{t,j}a_t(u)+\tfrac{1}{2}g_{t,j}^2\omega_t\right\}
	\distn{N}\left(v;b_{t,j}(u),\omega_t\right),	
\]
where $b_{t,j}(u)=a_t(u)+g_{t,j}\omega_t$ and $\distn{N}(v;b,\omega)$ denotes the normal density with mean $b$ and variance $\omega$ evaluated at $v$. Integrating over the $j$th interval gives the weight
\begin{equation} \label{eq:weights}
	w_{t,j}(u) = \exp\left\{\alpha_{t,j}+g_{t,j}a_t(u)+\tfrac{1}{2}g_{t,j}^2\omega_t\right\}
	\left[\Phi\left(\zeta^{+}_{t,j}\right)-\Phi\left(\zeta^{-}_{t,j}\right)\right],
\end{equation}
where $\zeta^{+}_{t,j} = (z_{t,j}-b_{t,j}(u))/\sqrt{\omega_t}$, 
$\zeta^{-}_{t,j} = (z_{t,j-1}-b_{t,j}(u))/\sqrt{\omega_t},$ and $\Phi$ is the standard normal distribution function. The normalizing constant is $\hat C_t(u)=\sum_{j\in\mathcal{A}_t}w_{t,j}(u)$, so that $q_t(\cdot \gvn u)$ is a mixture of $|\mathcal{A}_t|\leq G_t$ normal densities that share the variance $\omega_t$ and are each truncated to one interval. A draw is obtained by selecting the $j$th component with probability $w_{t,j}(u)/\hat C_t(u)$ and then drawing from $\distn{N}(b_{t,j}(u),\omega_t)$ truncated to $(z_{t,j-1},z_{t,j}]$ by inversion. Both steps are exact and together require $O(G_t)$ evaluations of $\Phi$.

The same expressions deliver the initial density: $q_1$ is obtained from \eqref{eq:weights} on replacing $a_t(u)$ by $\mu_1$ and $\omega_t$ by $\omega_1$, and $\hat C_1$ is the corresponding sum of weights. The dominating constant of Theorem~\ref{thm:exact} is therefore computed as a by-product of the backward recursion, at the cost of a single additional evaluation.

Constructing $\psi_t$ requires $f_t$ and its first derivative at the nodes. The first is available from \eqref{eq:weights} through \eqref{eq:f}. The second is available at no additional cost. Differentiation under the integral sign, justified in the proof of Lemma~\ref{lem:concave}, gives
\begin{equation} \label{eq:Cprime}
	\hat C_t'(u) = \frac{\beta_t}{\omega_t}\int_{\mathbb{R}}\left(v-a_t(u)\right)p_t(v \gvn u)\e^{\psi_t(v)}\,\di v,
\end{equation}
and dividing by $\hat C_t(u)$ gives
\begin{equation} \label{eq:derivative}
	\frac{\di}{\di u}\log\hat C_t(u) = \frac{\beta_t}{\omega_t}\left\{\mathcal{M}_t(u)-a_t(u)\right\},
	\qquad \mathcal{M}_t(u) = \Em_{q_t(\cdot \gvn u)}\left[h_t\right],
\end{equation}
so that $f_{t-1}'(u) = \ell_{t-1}'(u)+\beta_t\{\mathcal{M}_t(u)-a_t(u)\}/\omega_t$. The conditional mean of the twisted transition is a weighted average of truncated normal means,
\[
	\mathcal{M}_t(u) = \frac{\sum_{j\in\mathcal{A}_t}w_{t,j}(u)\mu_{t,j}(u)}{\sum_{j\in\mathcal{A}_t}w_{t,j}(u)},
	\qquad
	\mu_{t,j}(u) = b_{t,j}(u)+\sqrt{\omega_t}\,
	\frac{\varphi\left(\zeta^{-}_{t,j}\right)-\varphi\left(\zeta^{+}_{t,j}\right)}
	{\Phi\left(\zeta^{+}_{t,j}\right)-\Phi\left(\zeta^{-}_{t,j}\right)},
\]
where $\varphi$ is the standard normal density (not to be confused with the persistence parameter $\phi$ of the stochastic volatility model). Every quantity appearing here is already formed in computing $\hat C_t$, so the derivative recursion adds only $O(G_t)$ elementary operations per node.

\subsection{Choosing the Nodes} \label{ss:nodes}

The nodes affect only efficiency, so they should be placed where the smoothed states have mass. Our default implementation sets $v_{t,j}=m_t+w_js_t$, where $m_t$ is the $t$th element of the posterior mode of $\bh$, $s_t$ is the corresponding standard deviation under the Gaussian approximation at the mode, and $w_1<\cdots<w_G$ is a fixed grid on $[-4,6]$. The mode is computed by Newton's method, with each iteration requiring an $O(T)$ band solve. The standard deviations are the square roots of diagonal elements of the inverse of a tridiagonal precision matrix and are obtained in $O(T)$ operations using the selected inversion recursion of \citet*{TakahashiFaganChin1973}, rather than by forming the inverse at $O(T^2)$. The grid is asymmetric because in the stochastic volatility and duration models the observation log-density decays much faster to the left of the mode than to the right. The favorable orientation is model specific---it is reversed for Poisson counts and depends on $y_t$ for dynamic probit.

We refer to this mode-centered, fixed-range construction as the \emph{practical grid}. The scaling theorem below is proved for a different \emph{theoretical grid}, obtained from a companding rule chosen to permit uniform control of both curvature and proposal tails. This distinction concerns efficiency, not exactness: Theorem~\ref{thm:exact} applies to any finite nonempty node sets, including both grids. We use the simpler practical grid as the default implementation.

\subsection{How Many Nodes: Bounds and Scaling} \label{ss:scaling}

How many nodes are needed has two answers. We first give a finite-sample bound that applies to arbitrary node sets, including the practical grid of Section~\ref{ss:nodes}. We then introduce a particular companding grid for which that bound can be strengthened to a uniform $T/G^2$ scaling result. Write $S_T=-\sum_{t=1}^{T}d_t(h_t)$ for the accumulated envelope residual of a proposed path. Since $d_t\leq0$, $S_T\geq0$, and the acceptance probability of Theorem~\ref{thm:exact} is $\Em_q(\e^{-S_T})$. Jensen's inequality gives
$ \Em_q(\e^{-S_T})\geq \exp\{-\Em_q(S_T)\},$ and the expected residual admits the following finite-sample bound.

\begin{proposition}[Acceptance bound] \label{prop:rate}
Let the conditions of Theorem~\ref{thm:exact} hold, suppose $G_t\geq2$ for every $t$, let $I_t$ denote the interval spanned by $V_t$, let $\Delta_t$ be the largest spacing between adjacent nodes of $V_t$, and suppose $f_t'$ is Lipschitz on $I_t$ with constant $L_t$, which holds in particular if $f_t$ is twice differentiable there with $-f_t''\leq L_t$. Then
\[
	0\leq \psi_t(u)-f_t(u)\leq \frac{L_t\Delta_t^2}{8} \quad \text{for } u\in I_t,
\]
and consequently
\[
	\Pm(\text{accept}) \geq \exp\left\{-\sum_{t=1}^{T}\left[\frac{L_t\Delta_t^2}{8}+\Em_q(R_t)\right]\right\},
\]
where $R_t = \{\psi_t(h_t)-f_t(h_t)\}\bone\{h_t\notin I_t\}$ is the contribution from outside the node range.
\end{proposition}

The proof is in~\ref{app:proofs}. The proposition separates the two sources of envelope error. Inside the node range, the error is quadratic in the largest node spacing. If curvature is uniformly bounded and $\Delta_t=O(G^{-1})$, this contribution accumulates at rate $T/G^2$. The remaining term is the contribution from outside the node range. Thus, if the node ranges also provide sufficiently uniform tail coverage, taking $G$ proportional to $\sqrt{T}$ keeps acceptance bounded away from zero. At that rate the backward pass costs $O(TG^2)=O(T^2)$ and each proposal costs $O(TG)=O(T^{3/2})$, the latter being the marginal cost of an additional draw once the backward pass has been constructed.

Proposition~\ref{prop:rate} remains valid for the practical grid, but it does not impose the uniform control of the tail terms $\Em_q(R_t)$ that an asymptotic claim requires. The following assumption and theorem instead construct a companding grid for which both the curvature and tail contributions can be controlled uniformly. The result is stated in centered coordinates. For deterministic or data-dependent centers $c_{t,T}$, write $\tilde f_t(x)=f_t(c_{t,T}+x)$ and $X_t = h_t-c_{t,T}$ for the centered state. Both parts of the assumption are imposed over the single family of companding node sets constructed in Theorem~\ref{thm:scaling} from the constants $b$ and $W$ appearing below.

\begin{assumption}[Curvature and tail control] \label{as:scaling}
Each $\ell_t$ is twice continuously differentiable, and there exist constants $b>0$ and $\varsigma>2b/3$, a continuous function $W\geq1$ with $|\log W(x)-\log W(y)|\leq L_W|x-y|$ for all $x,y$, centers $c_{t,T}$, and finite nonnegative coefficients $A_{t,T}$ and $B_{t,T}$ with $\sup_TT^{-1}\sum_{t=1}^TA_{t,T}B_{t,T}<\infty$, such that for every $t$ and $T$, uniformly over the node sets of Theorem~\ref{thm:scaling}: (i) $0\leq-\tilde f_t''(x)\leq A_{t,T}\,W(x)$ for every $x$; and (ii) $\Em_q\big[\e^{\varsigma X_t^2}\big]\leq B_{t,T}$. The centers and the coefficients may be data dependent, in which case they are measurable functions of the data, defined together with $b$, $\varsigma$ and $W$ on a common probability-one event, every expectation under $q$ is read conditionally on the realized data, and the average condition on $A_{t,T}B_{t,T}$ is required to hold almost surely on that event.
\end{assumption}

Part (i) can be checked directly from the observation density. The backward recursion cannot generate uncontrolled curvature on its own: Lemma~\ref{lem:curvature} in \ref{app:proofs} shows that $0\leq-f_t''\leq-\ell_t''+\beta_{t+1}^2/\omega_{t+1}$, whatever the number and placement of the nodes at later dates. Part (ii) is the substantive condition, because it restricts an object generated by the algorithm rather than a primitive feature of the model. It requires the centered proposal marginals to be sub-Gaussian, uniformly in $t$, $T$ and $G$, at an exponent $\varsigma$ strictly above $2b/3$; the proof of Theorem~\ref{thm:scaling} shows that this is exactly what the approximation bound consumes. It is not automatic: the marginals of $q$ are outputs of the backward recursion, and the bound can fail for otherwise admissible node sets. The quantification is joint---the assumption asserts the existence of $b$, $\varsigma$ and $W$ such that the associated grids satisfy both (i) and (ii). In \ref{app:svscaling}, we construct such a triple for the stochastic volatility model and verify both conditions for the resulting grids.

\begin{theorem}[Scaling of the acceptance probability] \label{thm:scaling}
Let Assumptions~\ref{as:model} and \ref{as:scaling} hold, and at each date place $G\geq2$ nodes at $v_{t,j} = c_{t,T}+x_{j,G}$, where $x_{j,G} = \Lambda^{-1}\{(j-\frac{1}{2})/G\}$ and $\Lambda$ is the distribution function of the density proportional to $\{W(x)\e^{-bx^2}\}^{1/3}$. Then there is a constant $C$, depending only on $b$, $\varsigma$, $L_W$ and $W$, such that
\[
	-\log\Pm(\text{accept}) \leq \frac{C}{G^{2}}\sum_{t=1}^{T}A_{t,T}B_{t,T}.
\]
\end{theorem}

The proof is in~\ref{app:proofs}. Under the average bound of Assumption~\ref{as:scaling}, the right-hand side is $O(T/G^2)$. Hence $G$ proportional to $\sqrt{T}$ keeps the acceptance probability bounded away from zero, $G/\sqrt{T}\rightarrow\infty$ drives it to one, and $G$ proportional to $\sqrt{T/\log T}$ prevents it from decaying faster than polynomially. These conclusions apply to the companding grid in the theorem. The Monte Carlo results in Section~\ref{s:MC} show that the practical grid of Section~\ref{ss:nodes} displays the same $T/G^2$ scaling empirically, while delivering higher acceptance in the designs considered there. For the stochastic volatility model, the assumptions underlying the theoretical grid can be verified rather than imposed.

\begin{corollary} \label{cor:svscaling}
Consider the stochastic volatility model of Section~\ref{ss:model} with $\sigma_h^2>0$ and $|\phi|<1$, the state initialized from its stationary distribution. With centers $c_{t,T} = \log(1+y_t^2)$ and envelope $W(x) = 1+\e^{|x|}$, there are deterministic constants $b>0$ and $\varsigma>2b/3$, depending only on the model parameters, such that Assumption~\ref{as:scaling} holds almost surely under the data generating process. Consequently, for almost every data path there is a finite random constant $C_y$ for which
\[
	-\log\Pm(\text{accept})\leq C_y\,\frac{T}{G^{2}}
\]
at the node sets of Theorem~\ref{thm:scaling}, simultaneously for every $T$ and every $G\geq2$.
\end{corollary}

The square-root rate has a counterpart in the discretization filter of \citet{Farmer2021}, which replaces the continuous state process by a finite Markov chain and, for a scalar state, recommends growing the number of chain states in proportion to $\sqrt{T}$. The two prescriptions control different objects. There, refinement drives the error of approximating the continuous-state model to zero fast enough for the approximate likelihood to support valid inference; here, exactness holds at every finite node set by Theorem~\ref{thm:exact}, and refinement governs only the acceptance probability.

\subsection{The Algorithm} \label{ss:algorithm}

Everything the algorithm requires is now in place: the closed-form weights and derivatives of Section~\ref{ss:pieces} and the nodes of Section~\ref{ss:nodes}. A single backward pass forms the twists $\psi_1,\ldots,\psi_T$ and delivers the dominating constant $\hat C_1$ as a by-product, after which each proposal is one forward pass, accepted or rejected by \eqref{eq:accept}. It is summarized in Algorithm~\ref{alg:main}.

\begin{algorithm}[ht]
\caption{Tangent-twisted rejection sampling of the state path.}
\label{alg:main}
\begin{enumerate}
	\item \textit{Nodes.} Compute the posterior mode $m_t$ and the scale $s_t$, and set $v_{t,j}=m_t+w_js_t$.

	\item \textit{Backward pass.} Set $f_T(v_{T,j})=\ell_T(v_{T,j})$ and $f_T'(v_{T,j})=\ell_T'(v_{T,j})$. Then for $t=T,T-1,\ldots,2$:
	\begin{enumerate}
		\item[(a)] form $\psi_t$ from $\{v_{t,j},f_t(v_{t,j}),f_t'(v_{t,j})\}$, that is the slopes $g_{t,j}$, intercepts $\alpha_{t,j}$ and intersections $z_{t,j}$, using Lemma~\ref{lem:intervals} under strict concavity and the merging-and-pruning rule following that lemma otherwise;
		\item[(b)] evaluate $\hat C_t(v_{t-1,j})$ and $\mathcal{M}_t(v_{t-1,j})$ from \eqref{eq:weights} for $j=1,\ldots,G_{t-1}$;
		\item[(c)] set $f_{t-1}(v_{t-1,j})=\ell_{t-1}(v_{t-1,j})+\log\hat C_t(v_{t-1,j})$ and obtain $f_{t-1}'(v_{t-1,j})$ from \eqref{eq:derivative}.
	\end{enumerate}
	Finally form $\psi_1$ and compute $\hat C_1$.

	\item \textit{Proposal.} Draw $h_1\sim q_1$ and set $L=0$. Then for $t=2,\ldots,T$: compute the weights $w_{t,j}(h_{t-1})$ and $\hat C_t(h_{t-1})$ from \eqref{eq:weights}; update
	$L\leftarrow L+\ell_{t-1}(h_{t-1})+\log \hat C_t(h_{t-1})-\psi_{t-1}(h_{t-1})$; and draw $h_t$ from the piecewise normal density $q_t(\cdot \gvn h_{t-1})$. Finally update $L\leftarrow L+\ell_T(h_T)-\psi_T(h_T)$.

	\item \textit{Accept-reject.} Draw $U\sim\distn{U}(0,1)$. Return $\bh$ if $\log U\leq L$; otherwise return to step 3.
\end{enumerate}
\end{algorithm}

Step 1 costs $O(T)$ operations per Newton iteration, because the precision matrix is tridiagonal and each iteration is therefore a band solve, as in \citet{chan17}. Step 2 evaluates $G$ weights at each of $G$ nodes for each of $T$ periods and therefore requires $O(TG^2)$ evaluations of $\Phi$, where $G$ denotes the common number of nodes. Each proposal in step 3 requires $O(TG)$ evaluations, and the expected number of proposals per accepted draw is $\hat C_1/Z$. Steps 1 and 2 depend on the data and on the model parameters, but not on the number of draws required, so their cost is amortized when many draws are taken at a fixed parameter vector.

Theorem~\ref{thm:exact} guarantees that the right-hand side of \eqref{eq:accept} is nonpositive, but floating-point error can produce small positive values. Three precautions are therefore important. The mixture weights \eqref{eq:weights} are accumulated by log-sum-exp rather than by direct exponentiation and summation. Differences of normal distribution functions in \eqref{eq:weights} are evaluated using log distribution and survival functions, according to the signs of their arguments, since direct subtraction loses precision when both arguments lie far in the same tail. Finally, we monitor the largest value of $\sum_t d_t(h_t)$ over all proposals; a value materially above machine precision indicates a numerical or coding error and is investigated rather than silently truncated.

These diagnostics are used throughout the Monte Carlo study and empirical application. On deterministic grids containing the nodes, tangent intersections, interval midpoints and far-tail points, the residual never exceeds $4.5\times10^{-13}$, even in the demanding designs with $T=4{,}000$, $\phi=0.995$ or $\sigma_h^2=0.5$; positive values occur only at the level of double-precision roundoff. Direct exponentiation and subtraction of normal probabilities are numerically unstable in these designs, so the log-domain calculations described above are used throughout.

\section{Uses Beyond Exact State Simulation} \label{s:uses}

Beyond producing independent exact state draws, the tangent-twisted proposal has two uses that are particularly important for this paper. First, its bounded importance weights give a nonnegative unbiased likelihood estimator with guaranteed finite variance. Second, exact draws provide a benchmark for assessing approximate state smoothers. We develop these two uses first and then record several additional consequences of the same bound. The backward pass of Algorithm~\ref{alg:main} depends on the data and on the model parameters, so it must be rebuilt whenever the parameter vector changes, at a cost of $O(TG^2)$. Uses that vary the parameters are therefore most attractive when several state draws or likelihood evaluations are taken at each parameter value, or when computations can be parallelized.

Let $\bm{\theta}$ collect the parameters of the state and observation equations, and write $\gamma_{\bm{\theta}}$, $q_{\bm{\theta}}$, $\hat C_1(\bm{\theta})$ and $\ell_{t,\bm{\theta}}$ for the objects of Sections~\ref{s:method} and \ref{s:implementation} evaluated at $\bm{\theta}$. The normalizing constant
\[
    Z(\bm{\theta})
    =
    \int_{\mathbb{R}^T}\gamma_{\bm{\theta}}(\bh)\,\di\bh
    =
    p(\by\gvn\bm{\theta})
\]
is the observed-data likelihood, and $w_{\bm{\theta}}(\bh) =\gamma_{\bm{\theta}}(\bh)/q_{\bm{\theta}}(\bh)$ satisfies $ 0<w_{\bm{\theta}}(\bh)\leq\hat C_1(\bm{\theta})$ and $ \Em_{q_{\bm{\theta}}} \{w_{\bm{\theta}}(\bh)\}
 = Z(\bm{\theta})$ by Theorem~\ref{thm:exact} and Corollary~\ref{cor:is}. Write
$p_{\bm{\theta}}=Z(\bm{\theta})/\hat C_1(\bm{\theta})$ for the corresponding
rejection acceptance probability, reserving $p(\bm{\theta})$ for the prior
density.

\subsection{Likelihood Evaluation} \label{ss:likeuse}

For $\bh^{(1)},\ldots,\bh^{(M)}$ drawn independently from $q_{\bm{\theta}}$, define
\begin{equation} \label{eq:Zhat}
    \hat Z_M(\bm{\theta})
    =
    \frac{1}{M}
    \sum_{m=1}^{M}
    w_{\bm{\theta}}\big(\bh^{(m)}\big).
\end{equation}
The estimator is nonnegative and unbiased for $Z(\bm{\theta})$.
Corollary~\ref{cor:is} further gives
\[
    \frac{\Var\{\hat Z_M(\bm{\theta})\}}
         {Z(\bm{\theta})^2}
    \leq
    \frac{1/p_{\bm{\theta}}-1}{M}.
\]
Thus the likelihood estimator has finite variance at every parameter value, with a bound determined by a quantity that has a direct algorithmic interpretation. For example, an acceptance probability of $0.75$ implies relative variance at most $1/(3M)$. The bound is the exact relative variance of the cruder estimator $\hat C_1(\bm{\theta})M^{-1}\sum_m\bone\{\text{accept}^{(m)}\}$ constructed from the rejection indicators alone. Estimator~\eqref{eq:Zhat} is its Rao--Blackwellization \citep{CasellaRobert1996}, so the continuous weights should be retained rather than discarded at the rejection step.

Estimator~\eqref{eq:Zhat} can be used for simulated likelihood evaluation,
likelihood ratios and comparison of parameter values, in the manner of the
integrated likelihood estimators of \citet{CJ09}. The likelihood itself is
estimated without bias, although $\log\hat Z_M(\bm{\theta})$ is not an
unbiased estimator of the log likelihood. The same estimator also supports the marginal likelihood calculation used in the empirical application in Section~\ref{s:app}. Let $p(\bm{\theta})$ be the prior density and let $g$ be an importance density whose support contains that of $Z(\bm{\theta})p(\bm{\theta})$. If $\bm{\theta}^{(1)},\ldots,\bm{\theta}^{(N)}$ are drawn from $g$, and an independent estimate $\hat Z_M(\bm{\theta}^{(n)})$ is constructed at each draw, then
\begin{equation} \label{eq:marglike}
    \hat p(\by)
    =
    \frac{1}{N}
    \sum_{n=1}^{N}
    \frac{p(\bm{\theta}^{(n)})}
         {g(\bm{\theta}^{(n)})}
    \hat Z_M(\bm{\theta}^{(n)})
\end{equation}
is unbiased for $p(\by)$. Bounded state weights rule out an infinite variance arising from integration over the high-dimensional state path at a fixed $\bm{\theta}$; the remaining tail requirement concerns the lower-dimensional importance density $g$. \ref{app:data} gives a sufficient condition for the choice used in the application. 

More generally, because $\hat Z_M(\bm{\theta})$ is nonnegative and unbiased, it can also be used in pseudo-marginal parameter samplers \citep{Beaumont2003,AndrieuRoberts2009,ADH10}; \citet{LiuPlagborgMoller2023} illustrate the reach of this idea in economics, basing full-information Bayesian inference for heterogeneous agent models with unobserved aggregate states on a numerically unbiased likelihood estimator.

\subsection{Benchmarking State Approximations} \label{ss:benchmark}

Exact simulation is especially useful for benchmarking approximate state smoothers. Assessing a particle smoother, a Laplace or variational approximation, or a Metropolis--Hastings state sampler ordinarily requires a reference that is itself approximate, so the comparison measures the  difference between two approximation errors. Accepted paths from Algorithm~\ref{alg:main} instead provide draws from the intended smoothing distribution, leaving only ordinary Monte Carlo error. Posterior expectations of integrable functions can therefore be estimated from accepted paths without burn-in, thinning or autocorrelation correction. Alternatively, all proposal draws can be retained and used by ordinary or self-normalized importance sampling. The bounded weights prevent a small number of paths from dominating these estimates.

The rejection bound also gives a direct measure of the quality of the proposal
itself. Writing $\pi_{\bm{\theta}}(\bh) = p(\bh\gvn\by,\bm{\theta}),$ we have
$\pi_{\bm{\theta}}(\bh)/q_{\bm{\theta}}(\bh) \leq 1/p_{\bm{\theta}},$ with equality attained. Write $D_{\infty}$, $\mathrm{KL}$ and $\chi^{2}$ for the R\'{e}nyi divergence of order infinity, the Kullback--Leibler divergence and the chi-squared divergence. We have $D_{\infty} (\pi_{\bm{\theta}}\Vert q_{\bm{\theta}})  =  -\log p_{\bm{\theta}},$ $ \mathrm{KL} (\pi_{\bm{\theta}}\Vert q_{\bm{\theta}})    \leq   -\log p_{\bm{\theta}},$ and $ \chi^2  (\pi_{\bm{\theta}}\Vert q_{\bm{\theta}})   \leq  1/p_{\bm{\theta}}-1.$ Section~\ref{s:MC} uses exact simulation in precisely this benchmarking role.

\subsection{Further Consequences} \label{ss:furtheruses}

Several other consequences follow from the same bound. Algorithm~\ref{alg:main} can be used as an exact blocked state update in a posterior sampler, drawing $\bh$ from $p(\bh\gvn\by,\bm{\theta})$ before updating $\bm{\theta}$. This removes approximation and within-block Markov-chain error from the state update, although dependence remains through the parameter draws. Likewise, an independence sampler based on $q_{\bm{\theta}}$ satisfies
$P_{\bm{\theta}}(\bh,A)   \geq  p_{\bm{\theta}}\pi_{\bm{\theta}}(A)$ and is therefore uniformly ergodic \citep{tierney94,MengersenTweedie1996}. These observations are theoretical consequences rather than recommended implementations, since each draws only one state path per parameter value.

Finally, although $p_{\bm{\theta}}$ is unknown, it is itself an acceptance
probability and can therefore be estimated from bounded random variables.
A lower confidence bound for $p_{\bm{\theta}}$ immediately gives corresponding
upper bounds on $1/p_{\bm{\theta}}-1$ and $-\log p_{\bm{\theta}}$. This is
qualitatively different from diagnosing the variance of unbounded importance
weights, which may be driven by rare events absent from a finite simulation.

\section{Relation to Existing State Simulation Methods} \label{s:related}

The proposal developed here shares its architecture with several existing simulation smoothers, but differs in what it asks of the approximation error. In twisted methods, each transition density is multiplied by a function of the current state and renormalized, giving an importance density of the form \eqref{eq:proposal}. Lemma~\ref{lem:sep} is a property of this architecture: apart from an additive constant, the log target-to-proposal ratio is a sum of univariate functions of the states. Efficient importance sampling chooses Gaussian twists by backward regressions on simulated draws \citep{RichardZhang2007, LMRD2013}; the HESSIAN method of \citet{McCausland2012}, confined like the present paper to univariate states, matches derivatives of the target conditional log density through fifth order; and \citet{ScharthKohn2016} combine global twisting with particle resampling. A related family instead approximates the smoothing density around its mode \citep{SP97, DurbinKoopman1997, JK08, MMP11, chan17}. These methods seek a target-to-proposal ratio that varies little; the tangent construction instead controls its sign.

That distinction is a trade-off. The construction here matches only the level and first derivative of the backward function at each node, giving an error of order $\Delta^2$ in the node spacing, coarser than a fifth-order fit or a variance-minimizing regression. In exchange, every $d_t$ is nonpositive, so Lemma~\ref{lem:sep} turns the global dominating constant into a computable quantity. Importance sampling requires only support coverage and integrable weights; bounded weights, as in \citet{McCausland2012}, are a stronger sufficient condition that also guarantees finite variance. Rejection sampling requires more: the value of a finite dominating constant must be known. A period-by-period bound does not solve this problem, because multiplying the $T$ local bounds can lead to exponential deterioration when their average log gap is bounded away from zero. The tangent construction instead makes every coordinatewise residual supremum zero. Its univariate antecedents are \citet{Devroye1986} and \citet{GilksWild1992}; separability carries the envelope to the full state path.

The ensemble rejection sampler of \citet*{DeligiannidisDoucetRubenthaler2020} reaches exactness by a different route and covers a broader class of models, including nonlinear and non-Gaussian transitions. It performs rejection sampling on an extended particle space and obtains its cost guarantee from two-sided bounds on incremental weights. Under their Proposition~4, keeping acceptance bounded away from zero requires an ensemble size of order $T$, implying 
an $O(T^3)$ cost per exact draw. Those bounds do not apply in our setting because the state space is unbounded and the incremental weights have infimum zero; Theorem~\ref{thm:scaling} instead exploits concavity and Gaussian transitions to obtain the $T/G^2$ bound on the global envelope error.

Another class of methods approximates the measurement density rather than twisting the transitions. The auxiliary mixture sampler of \citet*{KSC98} for stochastic volatility replaces the measurement density by a finite mixture of normals, after which standard Gaussian simulation smoothers apply; extensions include leverage \citep{OCSN07}, stochastic volatility in mean \citep{HCO25}, and count and binomial models \citep{FruhwirthSchnatterWagner2006,FruhwirthSchnatterFruhwirthHeldRue2009}. These methods draw exactly from an approximating model, with any remaining discrepancy corrected, if desired, by reweighting. The construction here instead draws directly from the intended model under the log-concavity condition in Assumption~\ref{as:model}(ii).

The signed residual also distinguishes the method as an importance sampler. Corollary~\ref{cor:is} gives the explicit relative-variance bound $1/p-1$, which neither efficient importance sampling nor the HESSIAN method provides. Finite variance cannot in general be assumed for Gaussian proposals \citep{KoopmanShephardCreal2009}. For a Gaussian importance density with precision $\bQ_q$, the second moment is infinite whenever $\mathbf{d}'(2\bQ-\bQ_q)\mathbf{d}<0$ for a nonnegative direction $\mathbf{d}$ satisfying the tail condition of Section~\ref{ss:gaussian}. When $\bQ_q-\bQ$ is diagonal and the state coefficients $\beta_t$ are nonnegative, this reduces to an eigenvalue check on a tridiagonal matrix; for the Gaussian approximation at the mode, $\bQ_q$ is the matrix $\bK$ of Section~\ref{ss:gaussian} and the difference is diagonal by construction. The criterion fails in all stochastic volatility experiments and all six empirical fits below. The tangent weights are bounded by construction.

\section{Monte Carlo Evidence} \label{s:MC}

This section examines three aspects of the proposed method. First, we assess whether exact rejection sampling remains practical as the sample size grows and across different models. Second, we compare the resulting likelihood estimator with existing simulation-based likelihood methods. Third, we use exact simulation to benchmark commonly used approximate smoothers.

For the stochastic volatility design, $\mu=0$, $\phi=0.95$ and $\sigma^2_h=0.01$ unless stated otherwise. Acceptance probabilities are estimated by the Rao--Blackwellized form of the realized acceptance rate, $\widehat p = M^{-1}\sum_{m=1}^{M}\e^{-S_T^{(m)}}$, where $S_T=\sum_{t=1}^{T}\{\psi_t(h_t)-f_t(h_t)\}\geq0$ is the total envelope residual of a proposed path, as in Section~\ref{ss:scaling}, and $S_T^{(m)}$ is its value at the $m$th of $M$ proposals. All computations are in $\matlab$ on an otherwise idle desktop machine. The comparison in Section~\ref{ss:MClike} is calibrated by wall clock: the number of draws or particles each method receives is read off a measured cost model. The exact efficiency ratios are therefore hardware-specific; the common-budget protocol, rather than the numerical ratios, is the reproducible object. We run four checks to validate the implementation, each targeting a property the construction guarantees in theory.\footnote{First, on a dense grid, the residual $d_t$ never exceeds $9.1\times10^{-13}$ for any of the eight models in Table~\ref{tab:robust}, so the envelope is respected to numerical precision. Second, the likelihood estimates agree with deterministic quadrature for $T\leq50$. Third, when $y_t=0$ for every $t$, the observation log density is affine, the smoothing distribution is Gaussian, and the acceptance probability is $1.0$. Finally, at $T=60$, smoothed moments and quantiles from $50{,}000$ draws agree to within $0.01$ with those from the Gaussian-envelope sampler of Section~\ref{ss:gaussian}, which is exact and computationally feasible at that sample size.}

\subsection{Acceptance Performance} \label{ss:MCscaling}

Figure~\ref{fig:scaling} shows that the tangent construction delivers high acceptance probabilities even for long state paths. Panel~(a) compares it with the Gaussian envelope of Section~\ref{ss:gaussian}. For the Gaussian envelope, the acceptance probability decays approximately exponentially with $T$, falling to $\e^{-73.7}$ at $T=4{,}000$.\footnote{Probabilities this small are computed as follows: the two envelopes dominate the same unnormalized target, so the Gaussian acceptance probability equals $\widehat p\,\hat C_1\e^{-M_G}$ exactly, where $M_G$ is the Gaussian dominating constant of Section~\ref{ss:gaussian}, and only the factor $\widehat p$ is estimated. The same identity produces the Gaussian column of Table~\ref{tab:robust}.} With the tangent envelope and a fixed number of nodes, acceptance also declines with $T$, at a rate about one hundred times slower: at $T=4{,}000$ it is still $0.56$ with $G=81$. When instead $G$ is increased in proportion to $\sqrt{T}$, acceptance remains essentially constant at about $0.87$.

Panels~(b) and (c) show the source of this stability. The acceptance curves separate when plotted against $G$, but nearly coincide when plotted against $G/\sqrt{T}$. Thus the finite-sample behavior closely matches the $T/G^2$ envelope-error rate of Theorem~\ref{thm:scaling}: increasing the number of nodes at rate $\sqrt{T}$ prevents acceptance from deteriorating as the state dimension grows. Regressing $\log\Em_q(S_T)$ on $\log T$ and $\log G$ gives exponents of $1.00$ and $-2.03$, and the companding grid of Theorem~\ref{thm:scaling} gives $0.97$ and $-2.01$ at a uniformly lower level of acceptance.

\begin{figure}[H]
\centering
\includegraphics[width=\textwidth]{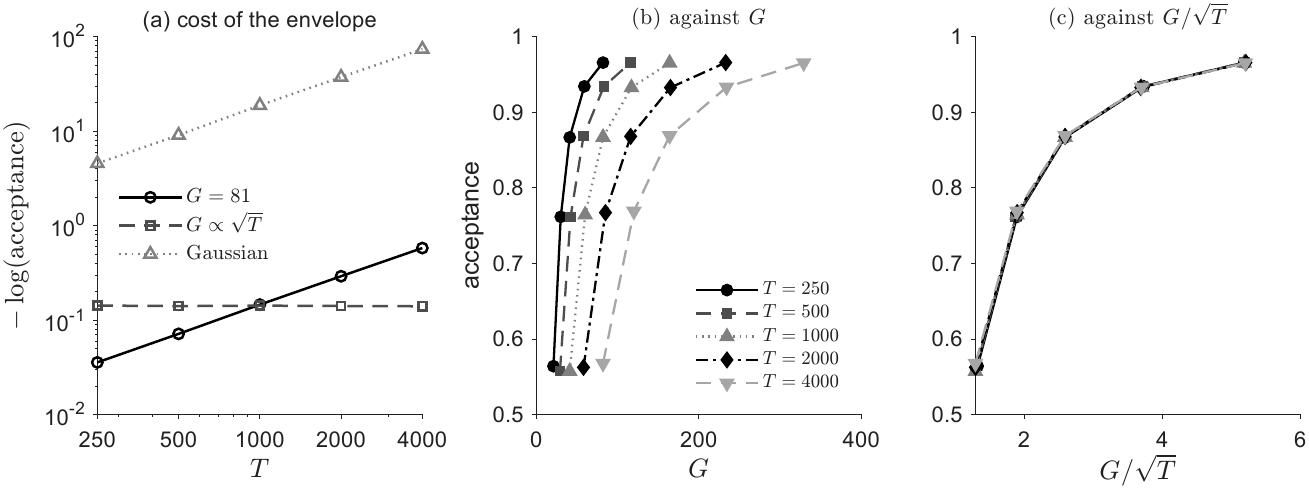}
\caption{Acceptance, sample size and node refinement for the stochastic volatility model. Panel (a) reports $-\log(\text{acceptance})$. Panels (b) and (c) report acceptance against $G$ and $G/\sqrt{T}$, respectively. Each point in panels (b) and (c) is based on 600 proposals.}
\label{fig:scaling}
\end{figure}

The high acceptance rates are not specific to stochastic volatility. Table~\ref{tab:robust} reports results for eight models at $T=500$, sharing the Gaussian AR(1) state equation and differing only in the observation density. Acceptance for the proposed method ranges from $0.811$ to $0.931$. In contrast, acceptance for the Gaussian envelope ranges from $1.1\times10^{-4}$ to $4.9\times10^{-66}$. The results are also stable across independently generated data sets: the standard deviation of the acceptance probability over twenty replications, reported in Table~\ref{tab:robust}, is below $0.004$ in seven of the eight designs, and $0.021$ for the Poisson. Hence the practical advantage of the tangent construction extends across the observation densities considered, rather than being specific to a particular stochastic volatility sample. 

\begin{table}[H]
\caption{Acceptance probabilities across eight models, $T=500$ and $G=81$. ``Across data sets'' is the standard deviation of acceptance over twenty independent data sets. The Student's $t$, duration and Poisson designs use $\sigma^2_h=0.05$, with $\mu=2$ for the Poisson, and the dynamic probit design $\phi=0.98$ and $\sigma^2_h=0.1$. The duration models have unit-mean exponential, gamma (shape $2$) and Weibull (shape $1.5$) errors. Every design uses the default node range $[-4,6]$ of Section~\ref{ss:nodes}, even where its orientation is unfavorable.}
\label{tab:robust}
\centering
\setlength{\tabcolsep}{4.5pt}
\begin{tabular}{lccc}
\hline\hline
Observation density & Tangent envelope & Across data sets & Gaussian envelope \\ \hline
Stochastic volatility         & 0.931 & 0.0010 & $1.1\times10^{-4}$ \\
\rowcolor{lightgray}
Stochastic volatility in mean & 0.918 & 0.0013 & $3.5\times10^{-6}$ \\
Student's $t$ errors & 0.896 & 0.0006 & $4.3\times10^{-10}$ \\
\rowcolor{lightgray}
Dynamic probit                 & 0.855 & 0.0034 & $1.2\times10^{-19}$ \\
Duration, exponential errors & 0.867 & 0.0004 & $3.2\times10^{-20}$ \\
\rowcolor{lightgray}
Duration, gamma errors & 0.851 & 0.0004 & $2.6\times10^{-30}$ \\
Duration, Weibull errors & 0.850 & 0.0004 & $7.7\times10^{-32}$ \\
\rowcolor{lightgray}
Poisson counts                 & 0.811 & 0.0206 & $4.9\times10^{-66}$ \\
\hline\hline
\end{tabular}
\end{table}

Acceptance also remains high when the stochastic volatility design is made more difficult: it is $0.885$ at $\phi=0.98$ and $0.855$ at $\phi=0.995$ as the state approaches a unit root, and $0.869$ and $0.835$ at $\sigma_h^2=0.1$ and $0.5$ as the innovation variance grows, while setting fifty observations exactly to zero, or inflating twenty observations by a factor of twenty, leaves it essentially unchanged.

\subsection{Likelihood Evaluation} \label{ss:MClike}

We next assess the tangent-twisted proposal for likelihood evaluation. We compare its unbiased likelihood estimator with efficient importance sampling (EIS), a Laplace--Gaussian importance sampler based on the Gaussian approximation at the posterior mode, and a bootstrap particle filter. The experiment uses $T=500$, $1{,}000$, $2{,}000$ and $4{,}000$. For each $T$, we generate five independent series and conduct four independent likelihood evaluations per method and series. Efficiency ratios are computed within each series and then averaged across series.

For likelihood evaluation we fix $G=121$ at every sample size rather than use the acceptance-based $G\propto\sqrt{T}$ rule of Section~\ref{ss:scaling}, since the relevant criterion is weight variance per unit of computing time. Figure~\ref{fig:like} compares the methods under a common wall-clock budget equal to the cost of one complete tangent-twisted evaluation with two hundred draws at a new parameter vector, including method-specific setup.\footnote{For EIS, setup consists of five backward regression sweeps using one thousand draws each; the Gaussian methods include the posterior-mode calculation. Comparator draw or particle counts are chosen to exhaust the same wall-clock budget. For the tangent-twisted and EIS estimators, the quantities in Figure~\ref{fig:like} are constructed from the relative variance $\mathrm{rv}_w$ of a single importance weight, estimated from separate long runs; panel~(a) uses the lognormal moment-matched value $\{\log(1+\mathrm{rv}_w/N)\}^{1/2}$. The particle-filter relative variance is estimated from independent replicated likelihood evaluations, since its likelihood estimate is a product of sequential averages rather than an average of independent path weights.}

\begin{figure}[H]
\centering
\includegraphics[width=0.92\textwidth]{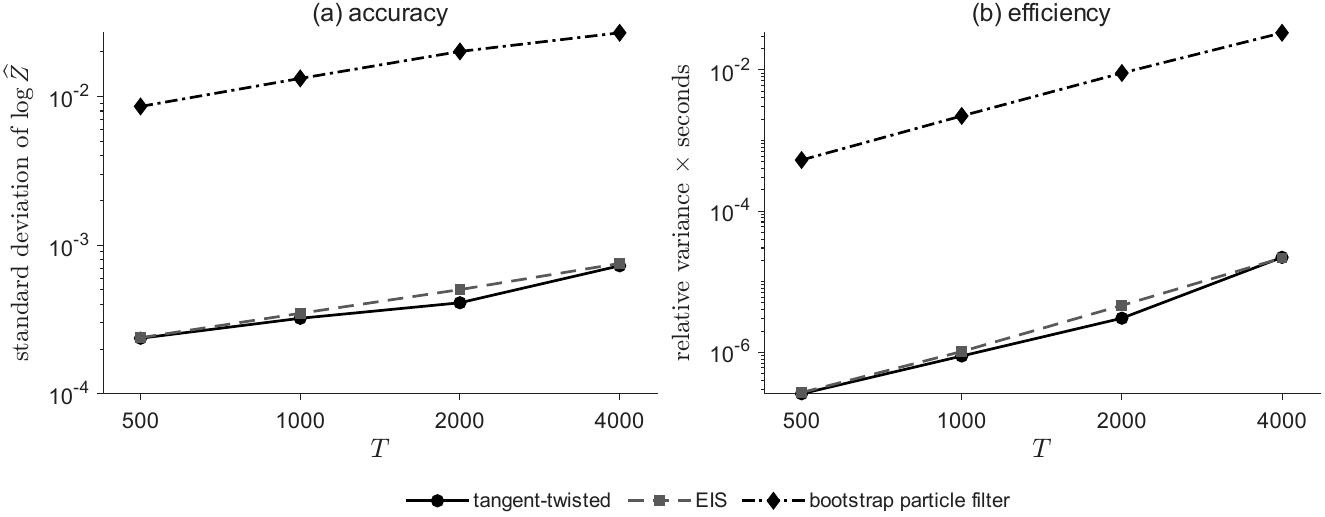}
\caption{Likelihood estimation for the stochastic volatility model at a fixed parameter vector under a common computing-time budget, with $G=121$. Each point averages over five independent series and four evaluations per method and series. The budget equals the cost of a tangent-twisted evaluation with two hundred draws, including setup; comparator draw or particle counts match it. Panel~(a) reports the implied standard deviation of $\log\widehat Z$, and panel~(b) the relative variance of $\widehat Z$ multiplied by computing time. Lower values indicate greater efficiency. The Laplace--Gaussian importance sampler is omitted because its weights have infinite variance in every design.}
\label{fig:like}
\end{figure}

The tangent-twisted estimator is substantially more efficient than the particle filter, by a factor of $1{,}900$ to $3{,}500$, and is comparable to EIS. The tangent proposal is more costly per draw because each transition is a mixture over $G$ pieces, but its lower weight variance largely offsets that cost, with neither method uniformly dominating across series and sample sizes. What distinguishes the tangent estimator is therefore not a systematic efficiency advantage over EIS, but its guarantee: its weights are bounded by a constant the algorithm computes, so Corollary~\ref{cor:is} ensures finite relative variance by construction. Consistent with this bound, the largest tangent weight exceeds the mean by less than $1.5\%$ at every sample size, compared with a factor of $55$ for EIS at $T=4{,}000$.

The Laplace--Gaussian comparison illustrates why this distinction matters. The finite-variance criterion of Section~\ref{s:related} fails for all twenty data sets, so neither quantity in Figure~\ref{fig:like} is defined for this proposal.\footnote{The smallest eigenvalue of $2\bQ-\bK$ ranges from $-0.41$ to $-0.26$. At $T=4{,}000$, the largest observed Laplace--Gaussian weight exceeds its mean by more than two orders of magnitude, against $1.013$ for the tangent proposal, and it grows with the number of draws, as the maximum of an infinite-variance quantity must. The same finite-variance criterion fails at all six empirical fits in Section~\ref{s:app}.} The tangent weights, by contrast, are bounded by construction.

\subsection{Benchmarking Approximate Smoothers} \label{ss:MCaccuracy}

A separate use of the method is to benchmark approximate state smoothers. Such assessments ordinarily compare one approximation with another; the tangent construction instead provides a reference for the intended smoothing distribution. We compare the Gaussian approximation at the posterior mode, EIS with self-normalized importance weights, and the genealogy smoother associated with a bootstrap particle filter. Each simulation-based method uses $8{,}000$ draws or particles, while the Gaussian approximation is evaluated in closed form. For greater precision, the reference uses $8{,}000$ tangent draws with self-normalized weights, whose effective sample size equals the nominal one to three decimal places.\footnote{\label{fn:refcheck}The reference is checked in three ways. First, the bound of Section~\ref{s:uses} gives $\chi^2(\pi_{\bm{\theta}}\Vert q_{\bm{\theta}})\leq1/p-1\leq0.339$ for $T\leq2{,}000$. Second, the largest realized weight exceeds the mean by at most $2.6\%$. Finally, the reference agrees with $4{,}000$ exact draws from Algorithm~\ref{alg:main}, with the largest absolute discrepancy in the smoothed mean ranging from $0.058$ to $0.070$ reference standard deviations. For comparison, two independent reference replicates differ by $0.045$ to $0.060$ under the same maximum-over-$t$ metric.} For each $t$, we examine the posterior mean and the $5\%$ quantile of $h_t$, with errors normalized by the reference marginal posterior standard deviation.

Table~\ref{tab:accuracy} shows three distinct patterns. EIS is extremely accurate: its root-mean-square error in the smoothed mean is $0.014$ to $0.017$ reference standard deviations, and its largest absolute error is $0.045$ to $0.058$, comparable to the Monte Carlo variation between independent reference replicates. Thus the case for the tangent construction is not that EIS gives materially inaccurate state estimates, but that it supplies exact draws and a computable global bound on the importance weights.

The Gaussian approximation exhibits a small but persistent approximation bias. Its root-mean-square error is about $0.08$ posterior standard deviations for the smoothed mean and $0.09$ for the lower quantile, with little change as $T$ increases. Because the Gaussian functionals are evaluated in closed form, these discrepancies reflect approximation rather than simulation error.

The genealogy smoother behaves differently: its error increases with $T$ and is concentrated near the beginning of the path, where particle ancestry has collapsed. At $T=2{,}000$, only about one percent of the original particles remain distinct over the first five percent of dates, and errors in the lower tail reach $0.72$ posterior standard deviations. The terminal filtering weights remain well behaved, with effective sample size at least $0.98$ of the nominal size, so the deterioration comes from genealogical collapse rather than poor filtering. This illustrates the benchmarking value of the exact sampler: approximation errors that are otherwise difficult to separate become directly measurable.

\begin{table}[H]
\caption{Error in the smoothed posterior mean and in the $5\%$ quantile of $h_t$, relative to the reference, in units of the reference marginal posterior standard deviation. RMSE is the root-mean-square error and Max the largest absolute error, both over $t$.}
\label{tab:accuracy}
\centering
\begin{tabular}{lrrrrrr}
\hline\hline
& \multicolumn{2}{c}{$T=500$} & \multicolumn{2}{c}{$T=1{,}000$} & \multicolumn{2}{c}{$T=2{,}000$} \\
& RMSE & Max & RMSE & Max & RMSE & Max \\ \hline
\multicolumn{7}{l}{\textit{Posterior mean}} \\
\rowcolor{lightgray}
Gaussian approximation & 0.075 & 0.109 & 0.079 & 0.116 & 0.078 & 0.123 \\
Efficient importance sampling & 0.016 & 0.046 & 0.014 & 0.045 & 0.017 & 0.058 \\
\rowcolor{lightgray}
Genealogy smoother & 0.051 & 0.158 & 0.082 & 0.240 & 0.094 & 0.394 \\
\multicolumn{7}{l}{\textit{$5\%$ quantile}} \\
\rowcolor{lightgray}
Gaussian approximation & 0.092 & 0.176 & 0.093 & 0.154 & 0.095 & 0.173 \\
Efficient importance sampling & 0.033 & 0.090 & 0.031 & 0.110 & 0.033 & 0.111 \\
\rowcolor{lightgray}
Genealogy smoother & 0.114 & 0.349 & 0.156 & 0.588 & 0.194 & 0.720 \\
\hline\hline
\end{tabular}
\end{table}

\section{Application: Stochastic Conditional Duration} \label{s:app}

We illustrate the method using the stochastic conditional duration model of \citet{BauwensVeredas2004}, providing an application outside stochastic volatility in which the choice of conditional duration distribution is an economically relevant model comparison. \citet{BauwensGalli2009} develop likelihood evaluation for this model using efficient importance sampling.

\subsection{Model Specifications and Data}

The stochastic conditional duration model is $y_t=\e^{h_t}\epsilon_t,$ 
where $\e^{h_t}$ is the latent expected duration, $h_t$ follows the affine Gaussian process in \eqref{eq:model}, and $\epsilon_t$ is a positive unit-mean error. We compare exponential, gamma and Weibull errors. The unit-mean normalization identifies the level of the state. Writing $u$ for the state, the corresponding observation log-densities are
\[
\ell_t(u)=-u-y_t\e^{-u}, \qquad
\ell_t(u)=-au-ay_t\e^{-u}, \qquad
\ell_t(u)=-ku-(y_t/\lambda)^k\e^{-ku},
\]
for the exponential, gamma with shape $a$, and Weibull with shape $k$, respectively, where $\lambda=1/\Gamma(1+1/k)$, up to terms independent of $u$. All three are concave, so Algorithm~\ref{alg:main} applies directly, and the gamma and Weibull families nest the exponential at unit shape.

We use Binance trade records for LINKUSDT and ALGOUSDT and construct seasonally adjusted volume durations for 15 May 2024; \ref{app:data} gives the data sources and construction details. Table~\ref{tab:appdata} shows substantial persistence and overdispersion in both series. The standard deviation exceeds the mean in each, with the greater dispersion for ALGOUSDT.

\begin{table}[H]
\caption{Adjusted volume durations, 15 May 2024. The seasonal range is the ratio of the largest to the smallest value of the periodic factor $\varphi_t$. Autocorrelations are of $\log y_t$.}
\label{tab:appdata}
\centering
\begin{tabular}{lrrrrrr}
\hline\hline
Pair & $T$ & Mean (s) & sd/mean & $\rho_1$ & $\rho_{20}$ & Seasonal range \\ \hline
LINKUSDT & 2{,}711 & 33.5 & 1.18 & 0.52 & 0.10 & 2.5 \\
\rowcolor{lightgray}
ALGOUSDT & 2{,}151 & 44.9 & 1.48 & 0.59 & 0.29 & 4.1 \\
\hline\hline
\end{tabular}
\end{table}

\subsection{Estimates and Model Comparison} \label{ss:appresults}

Each model is estimated by maximizing the unbiased likelihood estimator \eqref{eq:Zhat} using common random numbers; \ref{app:data} gives the implementation details. Table~\ref{tab:appfit} reports the estimates. Acceptance ranges from $0.78$ to $0.88$ across the six specifications. At all six estimates, the Gaussian approximation fails the finite-variance criterion of Section~\ref{s:related}, whereas the tangent weights remain bounded.

The model ranking is the same in both series: gamma is preferred to Weibull, which is preferred to exponential. The estimated gamma shapes, $0.642$ and $0.455$, are both well below the exponential benchmark of one and decline as the dispersion in Table~\ref{tab:appdata} rises. The error specification also materially affects the estimated state process: imposing exponential errors produces larger innovation variances and lower persistence, because variation that can instead be captured by the error distribution is forced into the latent state.

\begin{table}[H]
\caption{Simulated maximum likelihood estimates and the Bayesian information criterion (BIC). The exponential family has no shape parameter. Acceptance is the probability of Theorem~\ref{thm:exact} evaluated at the estimates with $G=188$.}
\label{tab:appfit}
\centering
\begin{tabular}{llrrrrrr}
\hline\hline
Pair & Family & $\mu$ & $\phi$ & $\sigma^2_h$ & Shape & Acceptance & BIC \\ \hline
LINKUSDT & Gamma & 3.21 & 0.991 & 0.015 & 0.642 & 0.856 & 22{,}741 \\
\rowcolor{lightgray}
& Weibull & 3.22 & 0.988 & 0.023 & 0.813 & 0.856 & 22{,}970 \\
& Exponential & 3.14 & 0.978 & 0.046 &---& 0.852 & 23{,}101 \\
\rowcolor{lightgray}
ALGOUSDT & Gamma & 3.21 & 0.981 & 0.072 & 0.455 & 0.883 & 17{,}167 \\
& Weibull & 3.24 & 0.969 & 0.175 & 0.606 & 0.880 & 17{,}367 \\
\rowcolor{lightgray}
& Exponential & 2.53 & 0.830 & 2.081 &---& 0.777 & 17{,}651 \\
\hline\hline
\end{tabular}
\end{table}

Because \eqref{eq:Zhat} is nonnegative and unbiased with bounded state weights, it can also be used in the marginal likelihood estimator \eqref{eq:marglike}. We use the priors
$\mu\sim\distn{N}(0,10)$, $(\phi+1)/2\sim\distn{B}(20,1.5)$, $\sigma_h^2\sim\distn{IG}(2.5,0.1)$,
and a standard normal prior for the log shape parameter. The parameter importance density is a heavy-tailed Student's $t$ centered at the simulated maximum likelihood estimate; \ref{app:data} gives the remaining details. Table~\ref{tab:appml} confirms the BIC ranking by a wide margin. Relative to gamma, the log marginal likelihood is lower by $101$--$114$ units for Weibull and by $178$--$249$ units for exponential. These differences are far larger than the numerical standard errors, which are at most $0.06$.

\begin{table}[H]
\caption{Log marginal likelihoods from \eqref{eq:marglike}, with numerical standard errors from $2{,}000$ parameter draws.}
\label{tab:appml}
\centering
\begin{tabular}{llrr}
\hline\hline
Pair & Family & $\log\hat p(\by)$ & NSE \\ \hline
LINKUSDT & Gamma & $-11{,}367.9$ & 0.06 \\
\rowcolor{lightgray}
& Weibull & $-11{,}481.7$ & 0.05 \\
& Exponential & $-11{,}546.4$ & 0.03 \\
\rowcolor{lightgray}
ALGOUSDT & Gamma & $-8{,}579.6$ & 0.05 \\
& Weibull & $-8{,}680.3$ & 0.05 \\
\rowcolor{lightgray}
& Exponential & $-8{,}828.9$ & 0.04 \\
\hline\hline
\end{tabular}
\end{table}

\section{Concluding Remarks and Future Research} \label{s:conclusion}

This paper develops an exact rejection sampler for non-Gaussian state space models with a scalar latent state, affine Gaussian dynamics, and a concave observation log-density. Backward twisting makes the log target-to-proposal ratio separable across dates, while tangent envelopes make every coordinatewise residual nonpositive and attain zero at the nodes. The resulting global dominating constant is therefore computable, attained, and smallest admissible for the proposal. Rejection sampling then produces independent draws from the intended joint smoothing distribution without burn-in, Markov-chain correction, or approximation-model error.

Exactness would be of limited practical value if acceptance deteriorated rapidly with the length of the state path. For the companding grid, the accumulated envelope error is \(O(T/G^2)\), so increasing the number of nodes at rate \(G\propto\sqrt{T}\) keeps the acceptance probability bounded away from zero. For stochastic volatility, the conditions underlying this result hold almost surely, while the simpler mode-centered grid used in practice displays the same scaling empirically. Importantly, exactness holds for every finite node set: refinement governs computational efficiency rather than the validity of accepted draws. The same proposal also yields a nonnegative unbiased likelihood estimator with bounded importance weights and therefore finite variance. In the numerical experiments, it is competitive with efficient importance sampling and provides an exact benchmark for assessing approximate state smoothers.

The main restriction of the present construction is the scalar latent state. For vector-valued states, the domination argument continues to apply, but the affine regions of the tangent envelope become polyhedra, making Gaussian probability evaluation and truncated-normal simulation the principal computational challenges. Developing practical methods for low-dimensional vector states is therefore a natural next step. Another direction is adaptive node placement that reduces the cost of the backward recursion while preserving the global bound. More broadly, the results suggest that one-sided approximations can be useful in latent-variable computation when a computable guarantee is more valuable than minimizing an unsigned approximation error.


\bigskip

\singlespacing
\bibliographystyle{econometrica}
\bibliography{var,exact-AR}

\onehalfspacing

\pagebreak

\appendix
\renewcommand{\thesection}{Appendix~\Alph{section}}

\section{Proofs} \label{app:proofs}

\begin{proof}[Proof of Lemma~\ref{lem:concave}]
The argument is by backward induction on $t$. Consider first $t=T$. By Assumption~\ref{as:model}(ii), $f_T=\ell_T$ is finite, concave and continuously differentiable, which is (i). A concave differentiable function lies below each of its tangent lines, so every term in the minimum in \eqref{eq:tangent} dominates $f_T$ and therefore $\psi_T\geq f_T$. At a node $v\in V_T$ the tangent line at $v$ takes the value $f_T(v)$ while every other tangent line is at least $f_T(v)$, so the minimum equals $f_T(v)$. Being a pointwise minimum of finitely many affine functions, $\psi_T$ is concave and piecewise linear. This is (ii).

Now suppose (i) and (ii) hold at some $t\geq 2$; we establish (iii) at $t$ and then (i) and (ii) at $t-1$. Since $\psi_t$ is a minimum of finitely many affine functions, there exist constants $c_0$ and $c_1$ with $\psi_t(v)\leq c_0+c_1v$ for every $v$. Hence
\[
	\hat C_t(u) = \int_{\mathbb{R}}p_t(v \gvn u)\e^{\psi_t(v)}\,\di v
	\leq \int_{\mathbb{R}}p_t(v \gvn u)\e^{c_0+c_1v}\,\di v
	= \exp\left\{c_0+c_1a_t(u)+\tfrac{1}{2}c_1^2\omega_t\right\}<\infty,
\]
using the moment generating function of the normal distribution, and $\hat C_t(u)>0$ because the integrand is strictly positive. The same argument applies to $\hat C_1$.

For log-concavity, write the integrand as $H_t(u,v) = p_t(v \gvn u)\e^{\psi_t(v)}$, so that
\[
	\log H_t(u,v) = -\tfrac{1}{2}\log(2\pi\omega_t)-\frac{\left(v-a_t(u)\right)^2}{2\omega_t}+\psi_t(v).
\]
Under Assumption~\ref{as:model}(i) the map $(u,v)\mapsto v-a_t(u)$ is affine, so its square is convex on $\mathbb{R}^2$ and the second term is jointly concave; the third term is concave in $v$ and constant in $u$, hence jointly concave. Therefore $H_t$ is log-concave on $\mathbb{R}^2$, and by the theorem of \citet{Prekopa1973} its marginal $\hat C_t$ is log-concave on $\mathbb{R}$, so $\log\hat C_t$ is concave.

For differentiability, write $\hat C_t(u) = (2\pi\omega_t)^{-1/2}C(a_t(u))$, where
\[
	C(a) = \int_{\mathbb{R}}\exp\left\{-\frac{(v-a)^2}{2\omega_t}+\psi_t(v)\right\}\di v
\]
is the convolution of a Gaussian density with $\e^{\psi_t}$. Because $\e^{\psi_t(v)}\leq\e^{c_0+c_1v}$, each $a$-derivative of the integrand is dominated, locally uniformly in $a$, by a normal density times a polynomial times $\e^{c_0+c_1v}$, which is integrable; differentiation under the integral sign is therefore justified to every order and $C$ is infinitely differentiable. Since $a_t$ is affine, $\hat C_t$ is infinitely differentiable in $u$, with derivative \eqref{eq:Cprime}. Since $\hat C_t(u)>0$, $\log\hat C_t$ is continuously differentiable, which completes (iii). Consequently $f_{t-1}=\ell_{t-1}+\log\hat C_t$ is finite, concave as the sum of two concave functions, and continuously differentiable, which is (i) at $t-1$; and (ii) at $t-1$ follows by the same tangent argument used in the base case.
\end{proof}

\begin{proof}[Proof of Theorem~\ref{thm:exact}]
Part (iii) of Lemma~\ref{lem:concave} shows that each $\hat C_t(u)$ and $\hat C_1$ is finite and strictly positive, so each conditional density in \eqref{eq:proposal} integrates to one and their Markov product $q$ is a joint probability density, which is (i).

By Lemma~\ref{lem:sep} and the definition of $d_t$, the ratio $\gamma(\bh)/q(\bh)$ equals $\hat C_1\exp\{\sum_t[f_t(h_t)-\psi_t(h_t)]\}$, and by part (ii) of Lemma~\ref{lem:concave} every summand is nonpositive. Hence $\gamma(\bh)\leq \hat C_1 q(\bh)$. Choosing any $\bh$ with $h_t\in V_t$ for all $t$ makes every summand zero, so the inequality is attained; by the second claim of Lemma~\ref{lem:sep} the supremum of the sum is the sum of the individual suprema, each of which is zero, so $\sup_{\bh}\gamma(\bh)/q(\bh)=\hat C_1$. This gives (ii).

For (iii), write $A(\bh) = \gamma(\bh)/\{\hat C_1q(\bh)\}$, which by (ii) satisfies $0<A(\bh)\leq 1$ and which is the acceptance probability implied by \eqref{eq:accept}. For any Borel set $B\subseteq\mathbb{R}^T$,
\[
	\Pm(\bh\in B,\ \text{accept}) = \int_{B}q(\bh)A(\bh)\,\di\bh = \frac{1}{\hat C_1}\int_{B}\gamma(\bh)\,\di\bh .
\]
Taking $B=\mathbb{R}^T$ gives $\Pm(\text{accept}) = Z/\hat C_1>0$. Dividing the two displays gives
$\Pm(\bh\in B \gvn \text{accept}) = \int_B\gamma(\bh)\,\di\bh/Z$, so the accepted draw has exactly the posterior distribution. Repeating the procedure with independent random numbers produces independent draws.
\end{proof}

\begin{proof}[Proof of Corollary~\ref{cor:is}]
The bound $w\leq\hat C_1$ is part (ii) of Theorem~\ref{thm:exact} and unbiasedness is immediate. Since $w$ is nonnegative and bounded by $\hat C_1$, $\Em_q(w^2)\leq \hat C_1\Em_q(w) = \hat C_1 Z$, so
$\Var_q(w)\leq \hat C_1 Z-Z^2$. Dividing by $Z^2$ and using $p = Z/\hat C_1$ gives the result.
\end{proof}

\begin{proof}[Proof of Lemma~\ref{lem:intervals}]
Let $s_{t,j} = \{f_t(v_{t,j+1})-f_t(v_{t,j})\}/(v_{t,j+1}-v_{t,j})$. By the mean value theorem $s_{t,j}=f_t'(\xi)$ for some $\xi\in(v_{t,j},v_{t,j+1})$, and strict concavity makes $f_t'$ strictly decreasing, so $g_{t,j}>s_{t,j}>g_{t,j+1}$. Substituting the definitions of $\alpha_{t,j}$ and $\alpha_{t,j+1}$ into the expression for $z_{t,j}$ and rearranging gives
\[
	z_{t,j} = v_{t,j}+(v_{t,j+1}-v_{t,j})\,\frac{s_{t,j}-g_{t,j+1}}{g_{t,j}-g_{t,j+1}},
\]
and the displayed ratio lies strictly between zero and one, so $v_{t,j}<z_{t,j}<v_{t,j+1}$. The intersection points are therefore increasing, and since the slopes are decreasing the minimum in \eqref{eq:tangent} is attained by the $j$th line precisely on $(z_{t,j-1},z_{t,j}]$.
\end{proof}

\begin{proof}[Proof of Proposition~\ref{prop:rate}]
Fix $u\in(v_{t,j},v_{t,j+1})$. Since $f_t'$ is $L_t$-Lipschitz on $I_t$, a first-order expansion with integral remainder about $v_{t,j}$ gives $f_t(u)\geq f_t(v_{t,j})+f_t'(v_{t,j})(u-v_{t,j})-L_t(u-v_{t,j})^2/2$, and the same expansion about $v_{t,j+1}$ gives the analogous inequality. Since $\psi_t(u)$ is at most the smaller of the two tangent lines,
$\psi_t(u)-f_t(u)\leq \frac{1}{2}L_t\min\{(u-v_{t,j})^2,(v_{t,j+1}-u)^2\}\leq \frac{1}{2}L_t(\Delta_t/2)^2$.
Splitting $\Em_q(S_T)$ into the contributions from inside and outside $I_t$ and applying Jensen's inequality gives the second display of the proposition.
\end{proof}

The remaining results in this appendix support Theorem~\ref{thm:scaling}. The first is the curvature bound cited after Assumption~\ref{as:scaling}; its content is that the backward recursion transmits at most $\beta_{t+1}^2/\omega_{t+1}$ of curvature from one date to the next, whatever the nodes at later dates.

\begin{lemma}[Curvature of a Gaussian twist] \label{lem:curvature}
Let the conditions of Lemma~\ref{lem:concave} hold and let $\mathcal{V}_t(u)$ denote the variance of the twisted transition density $q_t(\cdot \gvn u)$. Then $\mathcal{V}_t(u)\leq\omega_t$ and
\[
	(\log\hat C_t)''(u) = \beta_t^2\left\{\frac{\mathcal{V}_t(u)}{\omega_t^2}-\frac{1}{\omega_t}\right\}\in\left[-\frac{\beta_t^2}{\omega_t},0\right],
\]
so that, when $\ell_t$ is twice continuously differentiable,
\[
	0\leq-f_t''(u)\leq-\ell_t''(u)+\frac{\beta_{t+1}^2}{\omega_{t+1}}, \quad t<T,
	\qquad -f_T''(u)=-\ell_T''(u),
\]
for every number and placement of the nodes at dates $t+1,\ldots,T$. Both bounds on $(\log\hat C_t)''$ are sharp---the upper is attained and the lower is approached---so no downstream argument may assume strict inequality.
\end{lemma}

\begin{proof}
Write $a=a_t(u)$ and $C(a) = \int_{\mathbb{R}}\exp\{-(v-a)^2/(2\omega_t)+\psi_t(v)\}\,\di v$, so that $\hat C_t(u) = (2\pi\omega_t)^{-1/2}C(a_t(u))$. Since $\psi_t$ is bounded above by an affine function, differentiation under the integral sign is justified by Gaussian domination, and the standard tilting identities give $(\log C)'(a) = \{\mathcal{M}_t(u)-a\}/\omega_t$ and $\partial\mathcal{M}_t/\partial a = \mathcal{V}_t(u)/\omega_t$, hence $(\log C)''(a) = \mathcal{V}_t(u)/\omega_t^2-1/\omega_t$. As a function of $a$, $C$ is the convolution of a Gaussian density with the log-concave function $\e^{\psi_t}$ and is therefore log-concave by the theorem of \citet{Prekopa1973}; by the identity just displayed, this is equivalent to $\mathcal{V}_t(u)\leq\omega_t$, with no smoothness demanded of $\psi_t$. The lower bound follows from $\mathcal{V}_t(u)\geq0$, and the chain rule with $a_t'(u)=\beta_t$ gives the first display. The second follows from $f_t = \ell_t+\log\hat C_{t+1}$, whose second term is infinitely differentiable in $u$ as a Gaussian convolution, together with concavity of $f_t$ from Lemma~\ref{lem:concave}. Sharpness: when $\psi_t$ is a single tangent line ($G_t=1$), $\mathcal{V}_t(u)=\omega_t$ and the upper bound holds with equality, while a single sharp kink drives $\mathcal{V}_t(u)$ toward zero, so the lower bound is approached; it is never attained, since $q_t(\cdot \gvn u)$ has full support and $\mathcal{V}_t(u)>0$.
\end{proof}

The remaining result is the approximation bound behind Theorem~\ref{thm:scaling}: on the companding grid, curvature controlled by $W$ makes the tangent-envelope error of order $G^{-2}$ at every point, with a constant free of the curvature level, the function and the number of nodes. The bound is pointwise, so nothing is assumed about the law of the state; a sub-Gaussian moment then converts it into the expected error that Theorem~\ref{thm:scaling} needs.

\begin{lemma}[Companding approximation] \label{lem:quant}
Let $W\geq1$ be continuous with $|\log W(x)-\log W(y)|\leq L_W|x-y|$, let $b>0$, let $W_b = W\e^{-bx^2}$, and let $\Lambda$ be the distribution function of $\lambda = W_b^{1/3}/\int W_b^{1/3}$. For $G\geq2$ place nodes at $x_{j,G} = \Lambda^{-1}\{(j-\frac{1}{2})/G\}$, $j=1,\ldots,G$, so that the $j$th node lies in the $j$th cell $I_j = (\Lambda^{-1}\{(j-1)/G\},\Lambda^{-1}\{j/G\}]$ of $\lambda$-mass $1/G$. Let $g$ be concave and twice continuously differentiable with $0\leq-g''\leq A\,W$, and let $\psi_G$ be the tangent envelope of $g$ at the nodes. Then there is a $C$, depending only on $b$, $L_W$ and $W$, such that
\[
	0\leq\psi_G(x)-g(x)\leq\frac{CA}{G^2}\,W(x)^{1/3}\e^{2bx^2/3}
	\qquad\text{for every }x\in\mathbb{R}.
\]
Consequently, for every $\varsigma>2b/3$ there is a $C_\varsigma$, depending only on $b$, $\varsigma$, $L_W$ and $W$, such that any random variable $X$ satisfies
\[
	0\leq\Em\{\psi_G(X)-g(X)\}\leq\frac{C_\varsigma A}{G^2}\,\Em\big[\e^{\varsigma X^2}\big].
\]
\end{lemma}

\begin{proof}
All unsubscripted constants depend only on $(b,L_W,W)$. Write $x_0 = L_W/b$, $\vartheta(x) = \frac{1}{3}(2bx-L_W)$, and $d_G(x) = |x-x_{j,G}|$ for $x\in I_j$. The log-Lipschitz hypothesis gives, for $x_0\leq x<y$, the two-sided increment bounds
\begin{equation} \label{eq:increments}
	-\tfrac{1}{3}\{L_W+b(x+y)\}(y-x)
	\leq\log\lambda(y)-\log\lambda(x)
	\leq-\vartheta(x)(y-x),
\end{equation}
and symmetrically on the left, without differentiating $W$. Hence $\lambda$ is decreasing beyond $x_0$; integrating the upper bound, the $\lambda$-mass beyond $x\geq x_0$ is at most $\lambda(x)/\vartheta(x)$, and integrating the lower bound over $[x,x+1]$, it is at least $c\lambda(x)/(1+x)$.

\textit{Step 1: geometry.} Call a bounded segment either a bounded cell or the part of an extreme cell between its finite endpoint and its node; every bounded segment has $\lambda$-mass $1/G$ or $1/(2G)$, contains the node of its cell at an endpoint or in its interior, and the $\lambda$-mass beyond its outer endpoint is at least $1/(2G)$, this mass being exactly $1/(2G)$ when that endpoint is the extreme node and at least $1/G$ otherwise. Write $\mathcal K = [-x_0-2,x_0+2]$ and $\lambda_{\mathcal K} = \inf_{\mathcal K}\lambda>0$, and let $G_0\geq2$ be large enough that for $G>G_0$ both $1/G<\lambda_{\mathcal K}$ and the extreme nodes lie outside $\mathcal K$. For $2\leq G\leq G_0$ every bounded segment lies in the fixed compact set $[\Lambda^{-1}\{1/(2G_0)\},\Lambda^{-1}\{1-1/(2G_0)\}]$ and its length is at most that set's diameter. For $G>G_0$, a segment meeting the core $[-x_0,x_0]$ has length at most $4$, since a longer one would contain a unit interval inside $\mathcal K$, on which $\lambda\geq\lambda_{\mathcal K}>1/G$, contradicting a $\lambda$-mass of at most $1/G$; a segment outside the core in the right tail, with endpoints $x_L<x_R$, has $\lambda(x_R)\geq\vartheta(x_R)/(2G)$ by the tail-mass bound at $x_R$, and $\lambda$ decreasing gives length${}\leq(1/G)/\lambda(x_R)\leq2/\vartheta(x_R)$, which is bounded; symmetrically on the left. Moreover the length of a bounded segment times $|x+y|$, for $x,y$ in it, is uniformly bounded: for tail segments $|x+y|\leq2x_R$ and $2bx_R-L_W\geq bx_R$ for $x_R\geq x_0$, giving $2x_R\cdot2/\vartheta(x_R) = 12x_R/(2bx_R-L_W)\leq12/b$, while core-meeting segments have length at most $4$ and $|x+y|\leq2(x_0+4)$, giving $8(x_0+4)$. Write $\bar\Delta$ and $C_\Delta$ for the two resulting constants.

\textit{Step 2: the key inequality.} We claim
\begin{equation} \label{eq:keysup}
	\sup_{x\in\mathbb{R}}\ \lambda(x)\,d_G(x)\,\e^{L_Wd_G(x)/2}\ \leq\ \frac{C}{G}.
\end{equation}
On a bounded segment, $|\log W_b(x)-\log W_b(y)|\leq L_W\bar\Delta+bC_\Delta$ for $x,y$ in the segment, since $|x^2-y^2|\leq(y-x)|x+y|\leq C_\Delta$, so $\lambda = (W_b)^{1/3}/N$ varies by at most a constant factor $\e^{C_0}$ there; hence $\lambda(x)d_G(x)\leq\e^{C_0}\inf\lambda\cdot\text{length}\leq\e^{C_0}/G$, the product of infimum and length being at most the segment's $\lambda$-mass, while $\e^{L_Wd_G/2}\leq\e^{L_W\bar\Delta/2}$. On the unbounded part of the right extreme cell, $x = x_{G,G}+d$ with $d>0$: the reverse tail-mass estimate at $x_{G,G}$, whose tail mass is $1/(2G)$, gives $\lambda(x_{G,G})\leq C(1+x_{G,G})/G$, and \eqref{eq:increments} gives $\lambda(x)\leq\lambda(x_{G,G})\e^{-\vartheta(x_{G,G})d}$. Since $x_{G,G}\rightarrow\infty$ as $G\rightarrow\infty$, for $G$ large enough $\vartheta(x_{G,G})\geq L_W$, whence
\[
	\lambda(x)\,d\,\e^{L_Wd/2}
	\leq\lambda(x_{G,G})\,d\,\e^{-\vartheta(x_{G,G})d/2}
	\leq\lambda(x_{G,G})\,\frac{2}{\e\,\vartheta(x_{G,G})}
	\leq\frac{C}{G}\cdot\frac{1+x_{G,G}}{\vartheta(x_{G,G})}
	\leq\frac{C'}{G},
\]
because $(1+x)/\vartheta(x)\rightarrow3/(2b)$. The left tail is symmetric, and the remaining finitely many values of $G$ are absorbed into $C$, the supremum in \eqref{eq:keysup} being finite at each because $\lambda$ decays as a Gaussian while $d_G$ grows linearly.

\textit{Step 3: the envelope error.} Taylor's formula with integral remainder gives, for the node $v$ of the cell containing $x$ and its tangent $\tau_v$,
\[
	0\leq\psi_G(x)-g(x)\leq\tau_v(x)-g(x) = \int_{v}^{x}(x-z)\{-g''(z)\}\,\di z
	\leq\frac{A}{2}\,W(x)\,\e^{L_Wd_G(x)}\,d_G(x)^2,
\]
the last step bounding $\sup_{[x\wedge v,x\vee v]}W\leq W(x)\e^{L_W|x-v|}$ by the log-Lipschitz hypothesis. Since $\lambda = W_b^{1/3}/N$ with $N=\int W_b^{1/3}$ and $W_b = W\e^{-bx^2}$, the identity
\[
	W(x)\,\e^{L_Wd_G(x)}\,d_G(x)^2
	= N^2\,W(x)^{1/3}\e^{2bx^2/3}\left[\lambda(x)\,d_G(x)\,\e^{L_Wd_G(x)/2}\right]^{2}
\]
holds at every $x$, and \eqref{eq:keysup} bounds the bracket by $C/G$. This gives the first display. Nothing has been assumed about $X$: the envelope error is controlled pointwise, not merely on average, which is what separates the geometry of the grid from the law of the state.

For the second display, the log-Lipschitz hypothesis gives $W(x)^{1/3}\leq W(0)^{1/3}\e^{L_W|x|/3}$, and for every $\varsigma>2b/3$ there is a $C_\varsigma$ with $\e^{L_W|x|/3}\leq C_\varsigma\e^{(\varsigma-2b/3)x^2}$, a linear exponent being dominated by any positive quadratic one. Hence $W(x)^{1/3}\e^{2bx^2/3}\leq C_\varsigma\e^{\varsigma x^2}$, and taking expectations in the first display completes the proof.
\end{proof}

\begin{proof}[Proof of Theorem~\ref{thm:scaling}]
By Theorem~\ref{thm:exact}, $\Pm(\text{accept}) = \Em_q(\e^{-S_T})$ with $S_T = \sum_{t}\{\psi_t(h_t)-f_t(h_t)\}\geq0$, and Jensen's inequality gives $-\log\Pm(\text{accept})\leq\Em_q(S_T) = \sum_t\Em_q\{\psi_t(h_t)-f_t(h_t)\}$; only linearity of the expectation is used, so no independence across dates is required. Fix $t$. Recentering by $u = c_{t,T}+x$ preserves both tangency and pointwise minima, so $\psi_t(h_t)-f_t(h_t) = \tilde\psi_G(X_t)-\tilde f_t(X_t)$, where $X_t = h_t-c_{t,T}$ and $\tilde\psi_G$ is the tangent envelope of $\tilde f_t$ at $x_{1,G},\ldots,x_{G,G}$. The function $\tilde f_t$ is concave and twice continuously differentiable (Assumption~\ref{as:scaling} and the proof of Lemma~\ref{lem:curvature}) with $0\leq-\tilde f_t''\leq A_{t,T}W$ by Assumption~\ref{as:scaling}(i). Lemma~\ref{lem:quant} with $g=\tilde f_t$ and $A=A_{t,T}$ therefore bounds $\psi_t(h_t)-f_t(h_t)$ pointwise, and taking the expectation under $q$ gives $\Em_q\{\psi_t(h_t)-f_t(h_t)\}\leq(C_\varsigma A_{t,T}/G^2)\,\Em_q[\e^{\varsigma X_t^2}]\leq C_\varsigma A_{t,T}B_{t,T}/G^2$ by Assumption~\ref{as:scaling}(ii); only the marginal law of $X_t$ enters. Summing over $t$ completes the proof.
\end{proof}

\section{Verification for Stochastic Volatility} \label{app:svscaling}

This appendix proves Corollary~\ref{cor:svscaling}. Throughout, the model is the stochastic volatility model of Section~\ref{ss:model}: $y_t = \e^{h_t/2}\epsilon_t$ with $\epsilon_t\sim\distn{N}(0,1)$, $h_t = \mu+\phi(h_{t-1}-\mu)+\sigma_h\eta_t$ with $|\phi|<1$, initialized from the stationary distribution, so that $a_t(u) = \mu+\phi(u-\mu)$, $\omega_t=\sigma_h^2$ and $\omega_1 = \sigma_h^2/(1-\phi^2)$. The centers are $c_t = \log(1+y_t^2)$, the envelope is $W(x) = 1+\e^{|x|}$, which is log-Lipschitz with $L_W=1$ since $|(\log W)'| = \e^{|x|}/(1+\e^{|x|})<1$, and the nodes are those of Theorem~\ref{thm:scaling}: $v_{t,j} = c_t+x_{j,G}$ with $x_{j,G}$ the midpoint quantiles of $\lambda\propto\{W(x)\e^{-bx^2}\}^{1/3}$ for a small $b>0$ fixed in the proof of Lemma~\ref{lem:envelope}, which the paragraph following Assumption~\ref{as:scaling} explains is a joint rather than a circular quantification. All expectations $\Em_q$ are conditional on the data; the randomness of the data enters only at the final step.

Centering the state at $c_t=\log(1+y_t^2)$ removes the observation-specific scale from the curvature bound, which establishes part~(i) of the assumption by the direct calculation below. Part (ii) rests on two substantive facts, which are Lemmas~\ref{lem:score} and \ref{lem:envelope}: the mode of each twisted transition stays within a stationary, exponential-square-integrable distance of the center and contracts at rate $|\phi|$ in its argument; and each twisted transition is strongly log-concave, so that a stable drift recursion carries that mode control into a uniform sub-Gaussian moment bound on the marginals. Lemmas~\ref{lem:nodes} and \ref{lem:centers} supply the two inputs those facts require: a node near every center, and exponential-square moments for the centers themselves. Stationarity and ergodicity of the centers then convert that marginal bound into the average condition on $A_{t,T}B_{t,T}$ that Assumption~\ref{as:scaling} imposes.

Three conventions keep the argument uniform. Date~1 is treated as a generic date with a degenerate transition, $a_1(u)\equiv\mu$ and variance $\omega_1$, which is how $q_1$ is already constructed in Section~\ref{ss:pieces}. The substitution $(\phi,\sigma_h^2)\rightarrow(0,\omega_1)$ applies to statements about $q_1$ itself---its mode, its strong log-concavity modulus and its contraction in a conditioning argument that is absent---so that $\delta_1 = |c_1-\mu|$ and the date-1 contraction factor is zero. It does not apply to the backward derivative $f_1' = \ell_1'+(\log\hat C_2)'$, which is a statement about the date-2 transition and retains $(\phi,\sigma_h^2)$; the bound of Lemma~\ref{lem:score}(i) is stated so as to cover $t=1$ in that form. Since $\omega_1\geq\sigma_h^2$, any constraint of the form $\varsigma<c/\sigma_h^2$ is imposed in its binding form $\varsigma<c/\omega_1$. A random variable $X$ has a finite exponential-square moment if $\Em\e^{\varsigma X^2}<\infty$ for some $\varsigma>0$, and the exponent of Assumption~\ref{as:scaling}(ii) is one such $\varsigma$; every random constant below is obtained from the centers $c_t$ by finitely many additions of deterministic constants and geometrically weighted sums, each of which preserves finiteness of the sub-Gaussian Orlicz norm $\|\cdot\|_{\mathrm{sG}}$, and hence of that moment, by its triangle inequality, with no independence required \citep{Vershynin2018}, and we invoke this without further comment. Finally, $\epsilon$ is a tuning constant, unrelated to the observation noise $\epsilon_t$.

\subsection*{Curvature}

Write $\varpi_t = y_t^2/(1+y_t^2)\in[0,1)$, so that $\e^{-c_t} = 1-\varpi_t$ and $y_t^2\e^{-(c_t+x)} = \varpi_t\e^{-x}$. Hence
\[
	-\ell_t''(c_t+x) = \tfrac{1}{2}\varpi_t\e^{-x}\leq\tfrac{1}{2}\e^{|x|},
	\qquad
	\ell_t'(c_t) = -\tfrac{1}{2}+\tfrac{1}{2}\varpi_t = -\frac{1}{2(1+y_t^2)},
\]
so $|\ell_t'(c_t)|\leq\frac{1}{2}$, and by Lemma~\ref{lem:curvature},
\begin{equation} \label{eq:svcurv}
	0\leq-f_t''(c_t+x)\leq\mathcal{E}(x):=\tfrac{1}{2}\e^{|x|}+\frac{\phi^2}{\sigma_h^2}
	\leq A\,W(x),
	\qquad A = \max\left(\tfrac{1}{2},\frac{\phi^2}{\sigma_h^2}\right),
\end{equation}
for every $t$, $T$ and every placement of nodes at later dates. The centering matters here: the observation-specific factor $y_t^2$ disappears, leaving the deterministic coefficient $A$, so Assumption~\ref{as:scaling}(i) holds with $A_{t,T}\equiv A$. Under the data generating process $y_t\neq0$ almost surely, so $\varpi_t>0$ and $-\ell_t''>0$ everywhere; since $f_T=\ell_T$ and Lemma~\ref{lem:curvature} gives $(\log\hat C_{t+1})''\leq0$ whatever the number and placement of the nodes at later dates, $-f_t''\geq-\ell_t''>0$ at every date, so each $f_t$ is strictly concave and Lemma~\ref{lem:intervals} applies throughout the appendix. It remains to establish part (ii), which occupies the rest of it.

\subsection*{The Nodes and the Centers}

One property of the companding sets is used repeatedly, and is recorded here.

\begin{lemma}[Node near the center] \label{lem:nodes}
Let $\Lambda$ be the distribution function of $\lambda\propto W_b^{1/3}$, let $x_{j,G} = \Lambda^{-1}\{(j-\frac{1}{2})/G\}$ for $j=1,\ldots,G$ with $G\geq2$, and let $V_t = \{c_t+x_{j,G}: j=1,\ldots,G\}$. Then, uniformly in $t$, $T$ and $G$, the set $V_t$ contains a node within
\[
	\bar\varkappa = \left|\Lambda^{-1}\left(\tfrac{1}{4}\right)\right|\vee\left|\Lambda^{-1}\left(\tfrac{3}{4}\right)\right|
\]
of its center $c_t$, where $\bar\varkappa$ depends only on $b$, $L_W$ and $W$.
\end{lemma}

\begin{proof}
Write $u_j = (j-\frac{1}{2})/G$ and let $j^{*}$ be the largest $j$ with $u_j\leq\frac{1}{4}$, which exists because $u_1 = 1/(2G)\leq\frac{1}{4}$. If $u_{j^{*}}=\frac{1}{4}$ there is nothing to prove. Otherwise $j^{*}<G$, since $u_G = 1-1/(2G)\geq\frac{3}{4}$, and then $u_{j^{*}+1} = u_{j^{*}}+1/G<\frac{1}{4}+\frac{1}{2} = \frac{3}{4}$ while $u_{j^{*}+1}>\frac{1}{4}$ by maximality. Hence $x_{j^{*}+1,G}$ lies between $\Lambda^{-1}(\frac{1}{4})$ and $\Lambda^{-1}(\frac{3}{4})$, so $|x_{j^{*}+1,G}|\leq\bar\varkappa$ and the node $c_t+x_{j^{*}+1,G}$ is within $\bar\varkappa$ of $c_t$.
\end{proof}

Write $\bar c = \mathcal{E}(\bar\varkappa)\bar\varkappa$ for the slope-transfer constant and
\[
	\delta_t = \left|c_t-\mu-\phi(c_{t-1}-\mu)\right|
\]
for the centered innovation, so that $\delta_1 = |c_1-\mu|$ under the date-1 convention. The centers have finite exponential-square moments, for a reason that is specific to the logarithmic form of $c_t$ and is stated once here because the duration families of Section~\ref{s:app} use it again.

\begin{lemma}[Sub-Gaussian right tail of the centers] \label{lem:centers}
Let $h$ be normal with mean $\mu$ and variance $\omega$, let $\zeta>0$ be independent of $h$ with $\log\Em\zeta^{s}<\infty$ for every $s>0$ and $\log\Em\zeta^{s} = o(s^2)$ as $s\rightarrow\infty$, and set $c = \kappa^{-1}\log(1+\e^{\kappa h}\zeta)$ for a constant $\kappa>0$. Then $c\geq0$ and $\Em\e^{\varsigma c^2}<\infty$ for every $\varsigma<1/(4\omega)$. In particular the stochastic volatility centers $c_t = \log(1+y_t^2)$, which are of this form with $\kappa=1$, $\omega = \omega_1$ and $\zeta = \epsilon_t^2$, and with them the innovations $\delta_t$, have finite exponential-square moments and are stationary and ergodic.
\end{lemma}

\begin{proof}
Nonnegativity is immediate. Since $1+\e^{\kappa h}\zeta\leq2\max(1,\e^{\kappa h}\zeta)$,
\[
	0\leq c\leq\kappa^{-1}\log2+\left(h+\kappa^{-1}\log\zeta\right)^{+},
\]
so only the right tail of $S = h+\kappa^{-1}\log\zeta$ matters. This is what makes the argument work: $\log\zeta$ itself need not be sub-Gaussian, and for $\zeta = \epsilon^2$ it is not, its left tail being exponential; the positive part truncates that tail, and the centers inherit only the right one. For $s>0$, independence gives $\log\Em\e^{sS} = s\mu+\frac{1}{2}s^2\omega+\log\Em\zeta^{s/\kappa}$, whose last term is $o(s^2)$ by hypothesis, so Chernoff's bound at $s = x/(2\omega)$ gives
\[
	\Pm(S>x)\leq\exp\left\{-\frac{x^2}{2\omega}+\frac{x^2}{8\omega}+o(x^2)\right\}\leq\e^{-x^2/(4\omega)}
\]
for all $x$ large enough. Hence $S^{+}$, and with it $c$, has a finite exponential-square moment at every exponent below $1/(4\omega)$, and $\delta_t\leq c_t+|\phi|c_{t-1}+|\mu|(1+|\phi|)$ inherits one by the $\|\cdot\|_{\mathrm{sG}}$ triangle inequality. For the stochastic volatility model $\Em\big[(\epsilon^2)^{s}\big] = \Em|\epsilon|^{2s} = 2^{s}\Gamma(s+\frac{1}{2})/\sqrt{\pi}$, whose logarithm grows like $s\log s = o(s^2)$, and $h_t$ is stationary normal with variance $\omega_1$; stationarity and ergodicity are inherited from $(h_t,\epsilon_t)$.
\end{proof}

\subsection*{Mode Location and the Marginal Moments}

Lemmas~\ref{lem:score} and \ref{lem:envelope} below share one elementary concentration bound for the twisted transitions, centered at the mode.

\begin{lemma}[Mode-centered concentration] \label{lem:modetail}
Let $q(v)\propto\exp\{\psi(v)-(v-a)^2/(2\omega)\}$ on $\mathbb{R}$, where $\psi$ is concave and finite and $\omega>0$, let $r$ be the mode of $q$ and let $h\sim q$. Then $\Pm(|h-r|\geq z)\leq\e^{-z^2/(2\omega)}$ for every $z\geq0$, and consequently $\Em|h-r|\leq\sqrt{\pi\omega/2}$ and $\Em[\e^{\varsigma(h-r)^2}]\leq(1-2\varsigma\omega)^{-1}$ for every $0\leq\varsigma<1/(2\omega)$. The bounds depend on $\omega$ alone.
\end{lemma}

\begin{proof}
The log density $\Psi(v) = \psi(v)-(v-a)^2/(2\omega)$, taken without its normalizing constant, is coercive and strongly concave, so the mode is unique and $0\in\partial\Psi(r)$. For any $x$ and any supergradient $g\in\partial\Psi(x)$, the supergradient inequality for $\psi$ and the exact expansion of the quadratic give $\Psi(y)\leq\Psi(x)+g\,(y-x)-(y-x)^2/(2\omega)$; monotonicity of the superdifferential and $0\in\partial\Psi(r)$ supply a supergradient $g\leq0$ at every $x\geq r$. Taking $x = r+s-z$ and $y = r+s$ with $s\geq z\geq0$ yields $\Psi(r+s)\leq\Psi(r+s-z)-z^2/(2\omega)$, and integrating over $s\in[z,\infty)$,
\[
	\int_z^\infty\e^{\Psi(r+s)}\,\di s\leq\e^{-z^2/(2\omega)}\int_0^\infty\e^{\Psi(r+s)}\,\di s,
\]
so $\Pm(h-r\geq z)\leq\e^{-z^2/(2\omega)}\,\Pm(h\geq r)$; the mirror argument bounds the left tail, and the two add to the stated bound. The normalizing constant cancels in the ratio, so no bound on the height of the density is involved. Integrating the tail bound gives $\Em|h-r|\leq\int_0^\infty\e^{-z^2/(2\omega)}\di z = \sqrt{\pi\omega/2}$, and $\Em[\e^{\varsigma(h-r)^2}] = 1+\int_0^\infty2\varsigma z\,\e^{\varsigma z^2}\Pm(|h-r|\geq z)\,\di z\leq1+2\varsigma\int_0^\infty z\,\e^{-\{1/(2\omega)-\varsigma\}z^2}\di z = (1-2\varsigma\omega)^{-1}$.
\end{proof}

\begin{lemma}[Location of the twisted mode] \label{lem:score}
Let $r_t(u)$ denote the mode of $q_t(\cdot \gvn u)$ and write $F_{t,T} = |f_t'(c_t)|$, the second index recording the implicit dependence of $f_t$ on $T$. For the companding node sets the following hold.
\begin{enumerate}
	\item[(i)] $F_{t,T}\leq\bar F_t := C_0+C_1\sum_{j=1}^{\infty}|\phi|^{j-1}\delta_{t+j}$ for every $T\geq t$, with $C_0$ and $C_1$ depending only on $(\phi,\sigma_h,\bar\varkappa)$, and $\bar F_t$ is stationary with a finite exponential-square moment.
	\item[(ii)] $|r_t(u)-r_t(u')|\leq|\phi|\,|u-u'|$ for all $u,u'$.
	\item[(iii)] With $\varrho_t = 2\delta_t+\sigma_h^2(\bar F_t+\bar c)+\bar\varkappa$, which is stationary with a finite exponential-square moment,
	\[
		|r_t(c_{t-1}+x)-c_t|\leq\varrho_t+|\phi|\,|x|.
	\]
	At $t=1$ the date-1 convention gives $|r_1-c_1|\leq\varrho_1 := 2\delta_1+\omega_1(\bar F_1+\bar c)+\bar\varkappa$ with no contraction term; the prefactor $\omega_1$ comes from the first-order condition of $q_1$, while $\bar F_1$ and $\bar c$ retain $(\phi,\sigma_h^2)$ because they concern the date-2 transition.
\end{enumerate}
\end{lemma}

\begin{proof}
\textit{Claim 1: slope control at a mode.} Suppose $r-a = \sigma^2s$ for some $a\in\mathbb{R}$, some $\sigma^2>0$ and some supergradient $s\in\partial\psi_t(r)$. Let $v^{*}$ be the node of $V_t$ nearest $c_t$, so that $|v^{*}-c_t|\leq\bar\varkappa$ by Lemma~\ref{lem:nodes} and, by \eqref{eq:svcurv}, $|f_t'(v^{*})|\leq|f_t'(c_t)|+\bar c$. The slopes of $\psi_t$ are values of $f_t'$ at nodes and decrease from left to right, with $\psi_t'(v^{*+})\leq f_t'(v^{*})\leq\psi_t'(v^{*-})$; these hold with equality, because the interlacing of Lemma~\ref{lem:intervals} places $v^{*}$ strictly inside its active interval, so $\psi_t$ is differentiable at $v^{*}$ and the boundary case $r = v^{*}$ falls under either branch below. If $r\geq v^{*}$ then $s\leq\psi_t'(v^{*+})\leq|f_t'(c_t)|+\bar c$, while if in addition $s\leq0$ then $\sigma^2|s| = a-r\leq a-v^{*}\leq|a-c_t|+\bar\varkappa$; the case $r<v^{*}$ is symmetric. In every case
\begin{equation} \label{eq:slope}
	|s| \leq |f_t'(c_t)|+\bar c+\frac{|a-c_t|+\bar\varkappa}{\sigma^2}.
\end{equation}

\textit{Claim 2: the backward score recursion, giving (i).} At $t=T$, $F_{T,T} = |\ell_T'(c_T)|\leq\frac{1}{2}$. For $t<T$, the derivative recursion \eqref{eq:derivative} gives $f_t'(c_t) = \ell_t'(c_t)+(\phi/\sigma_h^2)\{\mathcal{M}_{t+1}(c_t)-a_{t+1}(c_t)\}$, where $\mathcal{M}_{t+1}(c_t)$ is the mean of $q_{t+1}(\cdot \gvn c_t)$. Let $r$ be the mode of that distribution, which exists and is unique because its log density is strongly concave with modulus $1/\sigma_h^2$ and coercive, and split $|\mathcal{M}_{t+1}(c_t)-a_{t+1}(c_t)|\leq|\mathcal{M}_{t+1}(c_t)-r|+|r-a_{t+1}(c_t)|$. For the first term, $q_{t+1}(\cdot \gvn c_t)$ has the form of Lemma~\ref{lem:modetail} with $\omega = \sigma_h^2$, so $|\mathcal{M}_{t+1}(c_t)-r|\leq\Em|h_{t+1}-r|\leq\sqrt{\pi/2}\,\sigma_h$. For the second term, the first-order condition produces $s\in\partial\psi_{t+1}(r)$ with $r-a_{t+1}(c_t) = \sigma_h^2s$, so \eqref{eq:slope} at date $t+1$ with $a = a_{t+1}(c_t)$, for which $|a-c_{t+1}| = \delta_{t+1}$, gives $|s|\leq F_{t+1,T}+\bar c+(\delta_{t+1}+\bar\varkappa)/\sigma_h^2$. Hence
\[
	F_{t,T}\leq\tfrac{1}{2}+|\phi|\Big(\frac{\sqrt{\pi/2}}{\sigma_h}+\bar c+\frac{\bar\varkappa}{\sigma_h^2}\Big)
	+\frac{|\phi|}{\sigma_h^2}\delta_{t+1}+|\phi|F_{t+1,T},
\]
in which the coefficient on $F_{t+1,T}$ is exactly $|\phi|$, because the error terms enter additively and never multiplicatively. Backward iteration gives (i), and Lemma~\ref{lem:centers} gives $\delta_t$ a finite exponential-square moment, so that the partial sums of $\bar F_t$ give $\|\bar F_t\|_{\mathrm{sG}}\leq C_0'+C_1\|\delta_t\|_{\mathrm{sG}}/(1-|\phi|)<\infty$ by monotone convergence; stationarity follows from that of $(h_t,\epsilon_t)$.

\textit{Claim 3: contraction and location, giving (ii) and (iii).} The first-order conditions produce supergradients $s\in\partial\psi_t(r_t(u))$ and $s'\in\partial\psi_t(r_t(u'))$ with $r_t(u)-a_t(u) = \sigma_h^2s$ and $r_t(u')-a_t(u') = \sigma_h^2s'$. Subtracting, and using $a_t(u)-a_t(u') = \phi(u-u')$,
\[
	r_t(u)-r_t(u') = \phi(u-u')+\sigma_h^2(s-s').
\]
Concavity of $\psi_t$ makes the superdifferential monotone, $(s-s')\{r_t(u)-r_t(u')\}\leq0$, so multiplying the display by $r_t(u)-r_t(u')$ gives $\{r_t(u)-r_t(u')\}^2\leq\phi(u-u')\{r_t(u)-r_t(u')\}$, which is (ii). That contraction reduces (iii) to $u = c_{t-1}$: there $r_t(c_{t-1})-c_t = \{a_t(c_{t-1})-c_t\}+\sigma_h^2s$ with $|a_t(c_{t-1})-c_t| = \delta_t$, and \eqref{eq:slope} with $|f_t'(c_t)|\leq\bar F_t$ gives $\sigma_h^2|s|\leq\sigma_h^2(\bar F_t+\bar c)+\delta_t+\bar\varkappa$. At $t=1$ the same computation with $a_1(u)\equiv\mu$ and variance $\omega_1$ gives the degenerate bound. The moment claim for $\varrho_t$ follows from (i).
\end{proof}

Centering the tail bound at the mode is not a matter of convenience. Centering instead at $a_t(h_{t-1})$, so as to keep the Gaussian factor of the transition density and bound $\e^{\psi_t}/\hat C_t$ separately, fails: by \eqref{eq:svcurv} the slope of $\psi_t$ at distance $d$ from the center can be of order $\e^{d}$, so that ratio can grow like $\exp(d\,\e^{d})$, which no Gaussian factor absorbs. The first-order condition at the mode replaces that slope by a supergradient, which \eqref{eq:slope} bounds, and it is about the mode that the concentration of Lemma~\ref{lem:modetail} is available.

\begin{lemma}[Uniform sub-Gaussian marginals] \label{lem:envelope}
There exist a deterministic $\varsigma>0$, depending only on the model parameters, a stationary ergodic sequence $K_t$ with $\Em K_t<\infty$, and an almost surely finite random variable $\Upsilon\geq1$ such that, uniformly in $T$ and $G\geq2$, the centered state $X_t=h_t-c_t$ satisfies $\Em_q[\e^{\varsigma X_t^2}]\leq\Upsilon K_t$. Since $\varsigma>2\varsigma/3$, Assumption~\ref{as:scaling}(ii) then holds at $b=\varsigma$ with $B_{t,T} = \Upsilon K_t$.
\end{lemma}

\begin{proof}
The constants are chosen in the order: the exponential-square exponent $\lambda_\varrho$ of $\varrho_t$ is fixed by Lemma~\ref{lem:score}(iii), and may be fixed uniformly in $b$ because $b$ enters $\varrho_t$ only through the additive constants $\bar\varkappa$ and $\sigma_h^2\bar c$, which change the moment's value but not the admissible exponent; then $\epsilon$ and $\varsigma$ are chosen explicitly in Step~2; finally $b=\varsigma$, which is admissible because $\varsigma>2\varsigma/3$. The quantification is therefore not circular: $\varsigma$ is determined by the model parameters alone, and $b$ is read off it.

\textit{Step 1: conditional sub-Gaussianity.} The log density of $q_t(\cdot \gvn u)$ is the concave $\psi_t$ plus a quadratic of curvature $-1/\omega_t$, so Lemma~\ref{lem:modetail} applies with $\omega = \omega_t$ and, writing $\Em_u$ for expectation under $q_t(\cdot \gvn u)$, gives $\Em_u[\e^{\varsigma_1(h_t-r_t(u))^2}]\leq(1-2\varsigma_1\omega_1)^{-1}$ for $\varsigma_1<1/(2\omega_1)$, uniformly in $u$, $t$, $T$, $G$; the constraint and the constant are imposed in their binding forms with $\omega_1\geq\sigma_h^2$, which covers the degenerate date~1 as well. Only a moment is required here, not a bound on the height of the conditional density; the two differ, because strong log-concavity alone leaves the peak unbounded, and Lemma~\ref{lem:modetail} never bounds it.

\textit{Step 2: drift recursion.} Recall $X_t = h_t-c_t$, the centered state of Section~\ref{ss:scaling}, and set $\xi_t = r_t(h_{t-1})-c_t$, so $X_t = \{h_t-r_t(h_{t-1})\}+\xi_t$ and $|\xi_t|\leq\varrho_t+|\phi||X_{t-1}|$, the degenerate date~1 reading $|\xi_1|\leq\varrho_1$. Set
\[
	\epsilon = \frac{1-|\phi|}{1+|\phi|},
	\qquad
	\bar\rho = (1+\epsilon)|\phi| = \frac{2|\phi|}{1+|\phi|}<1,
	\qquad
	c_\epsilon = (1+\epsilon)(1+1/\epsilon) = \frac{4}{1-\phi^2},
\]
all explicit in $\phi$ and well defined on the whole range $|\phi|<1$, including $\phi=0$, where $\epsilon=1$ and $\bar\rho=0$. Young's inequality twice gives $X_t^2\leq\bar\rho^2X_{t-1}^2+c_\epsilon\varrho_t^2+(1+1/\epsilon)\{h_t-r_t(h_{t-1})\}^2$. Conditioning on $h_{t-1}$, applying Step 1 with $\varsigma(1+1/\epsilon)<1/(2\omega_1)$, and then using Jensen's inequality for the concave map $x\mapsto x^{\bar\rho^2}$,
\[
	M_t := \Em_q\big[\e^{\varsigma X_t^2}\big]\leq C\,\e^{\varsigma c_\epsilon\varrho_t^2}\,M_{t-1}^{\bar\rho^2},
\]
where $C = \{1-2\varsigma(1+1/\epsilon)\omega_1\}^{-1}\geq1$ is the Step-1 constant at exponent $\varsigma(1+1/\epsilon)$. The recursion requires $\varsigma(1+1/\epsilon)<1/(2\omega_1)$ from Step~1, and the averaging below requires $\varsigma c_\epsilon/(1-\bar\rho^2)<\lambda_\varrho$; both hold at
\[
	\varsigma = \min\left\{\frac{1}{8\omega_1(1+1/\epsilon)},\ \frac{\lambda_\varrho(1-\bar\rho^2)}{2c_\epsilon}\right\}.
\]
At $t=1$ the transition is degenerate, so the contraction factor is $\bar\rho_1 = 0$ and the recursion terminates at $M_1\leq C\e^{\varsigma c_\epsilon\varrho_1^2}$. Under the stationary initialization the data sequence extends two-sidedly, and $\varrho_s$ is defined for every $s\in\mathbb{Z}$ from that extension; telescoping to $t=1$ and dominating the initial term,
\[
	M_t
	\leq\exp\Big\{\sum_{j\geq0}\bar\rho^{2j}\big(\log C+\varsigma c_\epsilon\varrho_{t-j}^2\big)\Big\}\,\Upsilon
	 =: M_t^{*}\,\Upsilon,
	\qquad \Upsilon = \e^{\varsigma c_\epsilon\varrho_1^2}\vee1,
\]
where $\Upsilon$ collects the geometrically discounted influence of date~1, is almost surely finite and is free of $t$; it cannot be absorbed into $M_t^{*}$, because the degenerate $\varrho_1$ of Lemma~\ref{lem:score}(iii) is not the stationary variable that enters the two-sided sum at lag $t-1$. Averaging with the weights $(1-\bar\rho^2)\bar\rho^{2j}$ inside the convex exponential,
\[
	\Em M_t^{*}\leq C^{1/(1-\bar\rho^2)}\,\Em\exp\Big\{\frac{\varsigma c_\epsilon}{1-\bar\rho^2}\varrho_t^2\Big\}<\infty
\]
by the choice of $\varsigma$. Setting $K_t = M_t^{*}$ gives $\Em_q[\e^{\varsigma X_t^2}] = M_t\leq\Upsilon K_t$ at every date, with $\Em K_t<\infty$; and $K_t$ is a measurable function of the stationary ergodic data sequence, hence stationary and ergodic. Since $\varsigma>2\varsigma/3$, Assumption~\ref{as:scaling}(ii) holds at $b=\varsigma$ with $B_{t,T}=\Upsilon K_t$.
\end{proof}

\begin{proof}[Proof of Corollary~\ref{cor:svscaling}]
Assumption~\ref{as:scaling}(i) holds with the deterministic $A_{t,T}\equiv A$ of \eqref{eq:svcurv} and $W(x) = 1+\e^{|x|}$, which satisfies $W\geq1$ and $L_W=1$; each $\ell_t$ is twice continuously differentiable. Assumption~\ref{as:scaling}(ii) holds with $b=\varsigma$ and $B_{t,T} = \Upsilon K_t$ by Lemma~\ref{lem:envelope}, uniformly over the companding node sets. For the average condition, Birkhoff's ergodic theorem gives $T^{-1}\sum_{t=1}^{T}A\Upsilon K_t\rightarrow A\Upsilon\,\Em K_0$ almost surely, which is finite because $\Upsilon$ is, and a convergent sequence is bounded, so $S:=\sup_TT^{-1}\sum_tA_{t,T}B_{t,T}$ is finite almost surely. Theorem~\ref{thm:scaling} then gives $-\log\Pm(\text{accept})\leq(C/G^2)\sum_tA_{t,T}B_{t,T}\leq C_yT/G^2$ with $C_y = CS$, which is finite on the same probability-one event and depends on neither $T$ nor $G$.
\end{proof}

The restriction to the companding sets is necessary, not merely convenient. Placing a single node at $c_t-L$ makes $\psi_t$ one tangent line of slope of order $\frac{1}{2}\varpi_t\e^{L}$; the twisted conditional mean then shifts from the prior mean by order $\sigma_h^2\e^{L}$, and no bound of the kind established above can hold uniformly over such configurations. This is why Assumption~\ref{as:scaling} is imposed over a single family of node sets rather than over all of them.

\subsection*{The Duration Families}

The centering device extends to the duration models of Section~\ref{s:app}. Table~\ref{tab:durfam} collects the three families: in each case the center is the logarithm of one plus the transformed duration, the observation-specific factor cancels from the curvature exactly as $y_t^2$ does in \eqref{eq:svcurv}, the score at the center is bounded by the shape parameter, and the hypotheses of Lemma~\ref{lem:centers} hold with the stated $(\kappa,\zeta)$ because $\log\Em\zeta^{s}$ grows like $s\log s = o(s^2)$.

\begin{table}[H]
\caption{The duration families of Section~\ref{s:app} in the notation of this appendix; the log-densities $\ell_t$ are those of Section~\ref{s:app}. In each row $-\ell_t''(c_t+z)\leq A\,W(z)$ for a deterministic $A$, and $\zeta$ is the variable of Lemma~\ref{lem:centers}, which applies with the stated $\kappa$. For the Weibull family $z_t = \{y_t\Gamma(1+1/k)\}^{k}$.}
\label{tab:durfam}
\centering
\begin{tabular}{lcccc}
\hline\hline
Family & $c_t$ & $-\ell_t''(c_t+z)$ & $W(x)$ & $(\kappa,\zeta)$ \\ \hline
Exponential & $\log(1+y_t)$ & $\frac{y_t}{1+y_t}\e^{-z}$ & $1+\e^{|x|}$ & $(1,\epsilon_t)$ \\
\rowcolor{lightgray}
Gamma, shape $a$ & $\log(1+y_t)$ & $a\frac{y_t}{1+y_t}\e^{-z}$ & $1+\e^{|x|}$ & $(1,\epsilon_t)$ \\
Weibull, shape $k$ & $k^{-1}\log(1+z_t)$ & $k^2\frac{z_t}{1+z_t}\e^{-kz}$ & $1+\e^{k|x|}$ & $\big(k,\{\epsilon_t\Gamma(1+\tfrac{1}{k})\}^{k}\big)$ \\
\hline\hline
\end{tabular}
\end{table}

We make two comments. The growth rate of $W_k(x)=1+\e^{k|x|}$ varies with $k$ in the Weibull case, so uniformity over an estimated shape requires restricting the parameter to $0<k\leq\bar k$ for some finite $\bar k$. The companding grid constructed from $\bar k$ then serves every $k\leq\bar k$, because the envelope functions $W_k$ are pointwise ordered, so the curvature bound of Assumption~\ref{as:scaling}(i) for $W_{\bar k}$ dominates that for each smaller $k$ and Lemma~\ref{lem:quant} applies once, at $\bar k$. The transition structure is identical to the stochastic volatility case, so the tail chain of Lemmas~\ref{lem:score}--\ref{lem:envelope} carries over step by step; we omit the repetition.

\section{Data and Implementation Details} \label{app:data}

This appendix provides the data, estimation and marginal likelihood details for the stochastic conditional duration application in Section~\ref{s:app}. The three subsections describe, in turn, how the raw trade records are converted into the adjusted volume durations in Table~\ref{tab:appdata}, how the exponential, gamma and Weibull models in Table~\ref{tab:appfit} are estimated by simulated maximum likelihood, and how the log marginal likelihoods in Table~\ref{tab:appml} are computed from \eqref{eq:marglike}.

\subsection*{Construction of the Durations}

\begin{sloppypar}
The data are the daily \texttt{aggTrades} archives published by Binance at \texttt{https://\allowbreak data.binance.vision/\allowbreak data/\allowbreak spot/\allowbreak daily/\allowbreak aggTrades/}. Each record reports a millisecond transaction time, price and quantity. We use LINKUSDT and ALGOUSDT for 13--17 May 2024 and estimate the models on 15 May.
\end{sloppypar}

We collapse records sharing a millisecond stamp into a single event. Zero durations between consecutive raw records account for $44\%$ of observations for the much more active BTCUSDT pair and $14\%$--$23\%$ for the pairs examined. After collapsing, all durations are positive. Volume durations are obtained by accumulating traded quantity until it crosses a threshold $V$ and recording the time between crossings. We set $V=\sum_t q_t/3{,}000$ on the estimation day. A trade that crosses several thresholds generates one event, so the realized counts in Table~\ref{tab:appdata} are somewhat below three thousand.

The periodic factor is $\varphi_t=\Em(y_t \gvn \text{time of day})$. We estimate it from the mean duration in each of $48$ half-hour bins, smooth the bin means using a circular five-point moving average with weights proportional to $(1,2,3,2,1)$, and normalize the factor to have mean one. The adjusted duration is $y_t/\varphi_t$. We estimate $\varphi_t$ by pooling all five days and apply it to 15 May. Pooling is what makes the intraday shape stable: the estimation day contributes about a fifth of the observations, and the factor is normalized to have mean one, so the level of activity on any single day has limited influence on the estimated profile. The model is estimated only on 15 May because trading activity changes sharply across the five-day window. For LINKUSDT, for example, the duration scale changes by more than a factor of three, which a time-of-day factor cannot absorb.

\subsection*{Estimation}

Each model is estimated by maximizing the unbiased likelihood estimator \eqref{eq:Zhat} over $\bm{\theta}=(\mu,\phi,\sigma_h^2)$ and the log shape parameter for the gamma and Weibull families. Terms in $\log p(y_t\gvn h_t)$ that do not depend on $h_t$ cancel from the state update but still depend on the shape parameter and must therefore be restored for shape estimation and model comparison. Without them, the profile likelihood increases without bound as the shape falls. At shape one the restored term is zero and both families reproduce the exponential model exactly, providing an implementation check. Common random numbers at every parameter value make the simulated objective smooth in $\bm{\theta}$.

Maximization uses Nelder--Mead with at most $1{,}500$ function evaluations, a budget no fit exhausts, initialized at $\mu=\log\bar y$, $\phi=0.9$, $\sigma_h^2=0.05$ and shape one. The search uses $G=61$ nodes and $M=100$ state draws per evaluation. Reported log likelihoods are reevaluated at $G=188$ and $M=500$, with $G/\sqrt{T}$ equal to $3.61$ and $4.05$ in the two series, consistent with the scaling rule in Section~\ref{ss:MCscaling}. The acceptance probabilities in Table~\ref{tab:appfit} are estimated at the parameter estimates using $G=188$, $300$ proposals and the Rao--Blackwellized estimator $\widehat p$ of Section~\ref{s:MC}. The state is unbounded in both directions, so the node range of Section~\ref{ss:nodes} is used unchanged. All computations are in $\matlab$ and single-threaded.

\subsection*{The Marginal Likelihood}

Section~\ref{ss:appresults} specifies the priors and the parameter importance density $g$. Bounded state weights do not by themselves guarantee finite variance of \eqref{eq:marglike}; the outer ratio $p(\bm{\theta})/g(\bm{\theta})$ must also be controlled. Corollary~\ref{cor:is} gives
\[
\Em\left\{\hat Z_M(\bm{\theta})^2 \gvn \bm{\theta}\right\}
\leq \left(1-\frac{1}{M}\right)Z(\bm{\theta})^2+\frac{1}{M}\hat C_1(\bm{\theta})Z(\bm{\theta}),
\]
so a sufficient condition for finite variance is
\begin{equation} \label{eq:finitevar}
\int \frac{p(\bm{\theta})^2}{g(\bm{\theta})}\left[\left(1-\frac{1}{M}\right)Z(\bm{\theta})^2+\frac{1}{M}\hat C_1(\bm{\theta})Z(\bm{\theta})\right]\di\bm{\theta}<\infty.
\end{equation}
Thus the tangent-twisted proposal controls integration over the state path, while $g$ controls the parameter-space tails.

We use $N=2{,}000$ draws of $\bm{\theta}$ and $M=25$ state draws at each value. The heavy-tailed Student's $t$ proposal is deliberate, because a Gaussian proposal need not provide sufficient tail coverage for condition~\eqref{eq:finitevar}. The proposal and numerical Hessian are constructed on the unconstrained parameterization $(\mu,\tanh^{-1}\phi,\log\sigma_h^2,\log\text{shape})$, with the corresponding Jacobians included in the prior. Working directly in $\phi$ is problematic because several estimates exceed $0.98$, so a finite-difference step can leave the stationary region on which the likelihood is defined.

Two safeguards apply to the scale, both imposed before the factor-of-two inflation. If the numerical curvature of the log posterior is negative, the corresponding proposal standard deviation implied by the second difference is capped at five. If the numerical curvature is nonnegative, the corresponding proposal standard deviation is set to one half. Together they ensure that a flat or noisy direction cannot collapse or inflate the proposal. Centering at the simulated maximum likelihood estimates rather than at the posterior mode is immaterial where the prior is uninformative relative to the likelihood, which is the case for all six fits.

These evaluations use $G=51$ nodes rather than the $188$ used for the reported log likelihoods. By Theorem~\ref{thm:exact} and Corollary~\ref{cor:is}, $\hat Z_M(\bm{\theta})$ is unbiased for any node set, so $G$ affects the variance of the inner estimator but not its validity. Corollary~\ref{cor:is} quantifies this contribution through $(1/p_{\bm{\theta}}-1)/M$. In these experiments, variation from the outer integration over $\bm{\theta}$ is more important, so the computational budget is better spent on parameter draws than on additional nodes.

\end{document}